\documentclass{article}
\usepackage{setspace}
\usepackage[margin=1.0in]{geometry}

\usepackage{amsmath,amssymb,mathtools,dsfont,enumitem,appendix,comment}
\usepackage{algorithm,algorithmic}
\usepackage{graphicx}
\usepackage{caption}
\usepackage{subcaption}
\usepackage{amsthm}
\usepackage{todonotes}
\newtheorem{corollary}{Corollary}

\newtheorem{example}{Example}

\newtheorem{theorem}{Theorem}[section] 
\newtheorem{definition}{Definition} 
\newtheorem{result}{Main Result}

\def\EE{{\mathbb{E}}}
\def\PP{{\mathbb{P}}}

\newcommand{\fall}{\,\forall\,}

\newcommand{\G}{{\cal G}}

\newcommand{\M}{{\cal M}}
\newcommand{\N}{{\cal N}}

\DeclareMathOperator*{\argmax}{arg\,max}

\usepackage{color}

\title{Re-incentivizing Discovery: Mechanisms for Partial-Progress Sharing in
Research\footnote{A shorter version of this paper appeared in ACM EC 2014 \cite{banerjee2014re}}}
\author{Siddhartha Banerjee\footnote{ORIE, Cornell University (was with MS\&E, Stanford University, during the time of this work)}
\and
Ashish Goel\footnote{Management Science and Engineering, Stanford University}
\and
Anilesh Kollagunta Krishnaswamy\footnote{Electrical Engineering, Stanford University}}
\begin{document}

\maketitle
\begin{doublespacing}
\begin{abstract}
An essential primitive for an efficient research ecosystem is
\emph{partial-progress sharing} (PPS) -- whereby a researcher shares information
immediately upon making a breakthrough. This helps prevent duplication of work;
however there is evidence that existing reward structures in research discourage
partial-progress sharing. Ensuring PPS is especially important for new online
collaborative-research platforms, which involve many researchers working on
large, multi-stage problems.

We study the problem of incentivizing information-sharing in research, under a
stylized model: non-identical agents work independently on subtasks of a large
project, with dependencies between subtasks captured via an acyclic
subtask-network. Each subtask carries a reward, given to the first agent who
publicly shares its solution. Agents can choose which subtasks to work on, and
more importantly, when to reveal solutions to completed subtasks. Under this
model, we uncover the strategic rationale behind certain anecdotal phenomena.
Moreover, for any acyclic subtask-network, and under a general model of
agent-subtask completion times, we give sufficient conditions that ensure PPS is
incentive-compatible for all agents. 

One surprising finding is that rewards which are approximately proportional to
perceived task-difficulties are sufficient to ensure PPS in all acyclic
subtask-networks. The fact that there is no tension between local fairness and
global information-sharing in multi-stage projects is encouraging, as it
suggests practical mechanisms for real-world settings. Finally, we show that PPS is necessary, and
in many cases, sufficient, to ensure a high rate of progress in research.
\end{abstract}

\section{Introduction}
\label{sec:intro}

Academic research has changed greatly since the Philosophical Transactions of
the Royal Society were first published 
in $1665$. However, the process of disseminating research has remained largely
unchanged, until very recently. In particular, journal publications have
remained to a large extent the dominant mode of research dissemination. However,
there is growing belief among researchers that existing systems are inefficient
in promoting collaboration and open information-sharing among researchers
\cite{Nosek2012a}. This is evident in the increasing use of platforms such as
ArXiv, StackExchange, the Polymath project, etc., for collaboration and research
dissemination \cite{Nielsen2011}.

Open information sharing in research \emph{reduces inefficiencies due to
researchers duplicating each other's work}. A researcher who has made a
breakthrough in a part of a large problem can help speed up the overall rate of
research by sharing this partial-progress -- however withholding this
information is often in her self-interest. Partial-progress sharing forms the
basis of new online platforms for collaborative research like the \emph{Polymath
project} \cite{polymath}, where mathematicians work on open problems with
discussions occurring online in full public view. Polymath and similar
platforms, though successful, are currently small, and participants often do not
receive any formal credit for their contributions. However mechanisms for
allocating credit for partial-progress are increasingly being recognized as an
important primitive for scaling such systems. For example, the results of the
first polymath were published under a generic pseudonym -- subsequently, the
Polymath wiki started publishing 
details regarding individual contributions, with grant acknowledgements. 

Incentivizing \emph{partial-progress sharing} in research forms the backdrop for
our work in this paper. In particular, we formally justify the following
assertions.
\begin{itemize}[nolistsep,noitemsep]
	\item[$\bullet$]\emph{Partial-progress sharing is critical for achieving
the optimal rate of progress in projects with many researchers.}
	\item[$\bullet$]\emph{It is possible to make partial-progress sharing incentive-compatible for all
researchers by carefully designing the rewards on the subtasks of a large project.}
\end{itemize}
Though these assertions appear intuitive, there are few formal ways of reasoning
about them. We develop a model of research dynamics with which we can study
these questions. 

\subsection{A Model for Strategic Behavior in Collaborative Projects}
\label{ssec:model}

We develop a stylized model for dynamics and information-sharing in research.
For brevity, we refer to it as the \emph{Treasure-hunt game}, as it is
reminiscent of a treasure-hunt where agents solve a series of clues towards some
final objective, and solving certain sets of clues `unlocks' new clues.

We consider $n$ non-identical agents working on a large project, subdivided into
$m$ inter-dependent subtasks. The dependencies between subtasks are captured via
a \emph{subtask-network} -- a directed acyclic graph, where edges represent
subtasks. A subtask $u$ becomes \emph{available} to an agent once she knows the
solution to \emph{all} the predecessor subtasks -- those from which subtask $u$
is reachable in the subtask-network. The entire project is completed once all
subtasks are solved. 

\begin{figure}[ht!]
\centering
        \begin{subfigure}[b]{\textwidth}
        \centering
                \includegraphics[scale=0.8]{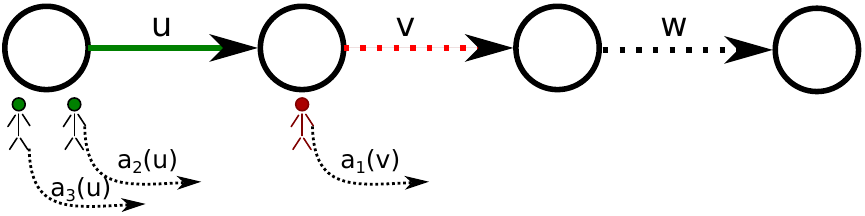}
                \caption{The Treasure-hunt game on a line}
        \end{subfigure}   
        \begin{subfigure}[b]{\textwidth}
        \centering
        \vspace{20pt}
                \includegraphics[scale=0.8]{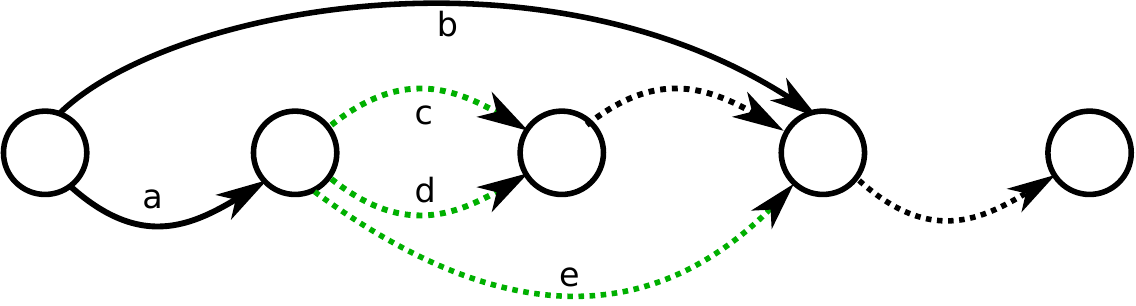}
                \caption{More complex acyclic subtask-network}
        \end{subfigure}
\caption[Treasure-Hunt Game]
{Examples of the Treasure-Hunt game: (a) Linear subtask-network -- Agent $1$ has solved (but not shared) subtask $u$ and can start subtask $v$. The rest must first finish subtask $u$. (b) More complex acyclic subtask-network. Solid
subtasks ($a,b$) are completed, while green subtasks ($c,d,e$) are available for
solving.}
\label{fig:thunt}
\end{figure}

For example, a project consisting of subtasks arranged in a line requires agents
to execute the subtasks in order starting from the first, and moving to the next
subtask only after completing the preceding one. On the other hand, a network
consisting of two nodes with parallel edges between them represents a project
with complementary subtasks -- each can be solved independently, but all
subtasks need to be solved to complete the project. More examples are given in
figure \ref{fig:thunt}.

We assume that \emph{agents work on subtasks independently}. If agent $i$
focusses her complete attention on subtask $u$, then she solves it after a
random time interval, drawn from an exponential distribution $Exp(a_i(u))$ --
the rate parameter $a_i(u)$ corresponds to the aptitude of agent $i$ for task
$u$. While most of our results are derived for a general aptitude matrix, we
obtain sharper results in the case of separable-aptitudes (SA):

\begin{definition}[\textbf{Separable-Aptitudes (SA) Model}]
\label{def:multiapt}
The agent-task aptitudes decompose as $a_i(u)=a_i\cdot s_u$, where:
\begin{itemize}[nolistsep,noitemsep] 
\item[$\bullet$] $a_i$ is the innate \emph{ability} of agent $i$, which is
independent of the tasks.
\item[$\bullet$] $s_u$ is the \emph{simplicity} of task $u$, which is
independent of the agents.
\end{itemize}
\end{definition}

\noindent We assume that agents know the complete aptitude matrix -- however, as
we discuss later, our results hold under much weaker assumptions. If an agent
has more than one subtask available for solving, then we assume she can split
her effort between them. 

Subtasks are associated with rewards, given to the first agent to \emph{publicly
share} its solution -- this captures the standard norm in academia, sometimes
referred to as the \emph{Priority Rule} \cite{Strevens2003}. We assume that the
\emph{rewards are fixed exogenously}. 

Upon solving a subtask, an agent can choose when to share the solution. Sharing
without delay helps speed up the overall rate of research; we formalize this as
follows:

\begin{definition}[\textbf{Partial-Progress Sharing (PPS) Policy}]
\label{def:pps}
An agent $i$ is said to follow the \emph{partial-progress sharing} policy if,
upon solving a subtask, she immediately shares the solution with all other
agents. 
\end{definition}

\noindent Agents can alternately defer sharing solutions to an arbitrary future
time, and continue working on downstream subtasks, while others remain stuck
behind. For example, an agent may choose to share her partial-progress after a
fixed delay, or until she solves the next subtask and can claim rewards for
multiple subtasks together. 
 
 The Treasure-hunt game captures the strategic behavior of agents with respect to
 sharing information. The strategy space of agents in the game consists of: $(i)$
 the delay in announcing the solution to a subtask, and $(ii)$ choice of subtask
 when multiple are available. On the other hand, for the social planner, there
 are two natural objectives:
 \begin{enumerate}[nolistsep,noitemsep]
 \item To \emph{ensure partial-progress sharing} -- choosing reward-vectors in a
 way that agents are incentivized to publicly share solutions without any delay.
 \item To \emph{minimize the makespan} -- ensuring the project is completed in
 the least time. 
 \end{enumerate}

\subsection{Two Warmup Examples}
\label{ssec:examples}

To understand the above objectives, we first present two examples. Our first
objective is to allocate rewards to incentivize PPS -- the following example
shows that even in simple subtask-network, the rewards affect agent strategies
in a non-trivial way:

\begin{example}
\label{eg:thunt}
Given a linear subtask-network under the SA model, with two subtasks
$\{p,q\}$ ($p$ preceding $q$) with simplicities $\{2s,s\}$, and two agents
$\{1,2\}$ with abilities $\{a,a\}$. If the rewards satisfy $4R_p>R_q$, then
for both agents, following PPS is a dominant strategy. Else, if the
rewards satisfy $4R_p<R_q$, then for both agents \emph{withholding
partial-progress} is a dominant strategy.
\end{example}

\begin{proof}

Suppose the strategy space of an agent $i$ is $S_i=\{P,\overline{P}\}$, where
$P$ corresponds to agent $i$ following PPS, and $\overline{P}$ to agent $i$
sharing solutions \emph{only after solving both} tasks\footnote{Though
sufficient for this example, we note that this is a simplified strategy space --
more generally, after solving an open subtask, an agent can choose to share the
solution after any arbitrary delay.}. For each $(x,y) \in S_1 \times S_2$, let
$U^i_{s_1,s_2}$ to be the payoff to agent $i \in \{1,2\}$ with agent $1$ opting
for strategy $s_1$ and agent $2$ opting for $s_2$. By symmetry (because both
agents are identical), we have
$U^1_{P,P}=U^2_{P,P}=U^1_{\overline{P},\overline{P}}=U^2_{\overline{P}}$,
$U^1_{P,\overline{P}}=U^2_{\overline{P},P}$ and
$~U^1_{\overline{P},P}=U^2_{P,\overline{P}}$. Further, since it is a constant
sum game with total reward $R_p+R_q$, we have:
$U^1_{P,P}+U^2_{P,P}=U^1_{P,\overline{P}}+U^2_{\overline{P},\overline{P}}=U^1_{P
,\overline{P}}+U^2_{P,\overline{P}}=U^1_{\overline{P},P}+U^2_{\overline{P},P}
=R_p+R_q$.
\noindent Furthermore, considering the strategy profile $(\overline{P},P)$, we
have: 
\begin{align*}
U^1_{\overline{P},P}&= \PP\left[\mbox{$1$ solves $p$
first}\right]\cdot\PP\left[\mbox{$1$ solves $q$ before $2$ claims $p$}|\mbox{$1$
solved $p$ first}\right]\cdot(R_p + R_q) \\
&+\PP\left[\mbox{$1$ solves $p$ first}\right]\cdot\PP\left[\mbox{$2$ claims $p$
before $1$ solves q}|\mbox{$1$ solved $p$
first}\right]\cdot\left(\frac{R_q}{2}\right) \\
&+ \PP\left[\mbox{$2$ solves $p$ first}\right]\cdot\left(\frac{R_q}{2}\right)\\
&= \frac{1}{2}\cdot\frac{sa}{sa + 2sa}\cdot(R_p + R_q) + \frac{2sa}{sa +
2sa}\cdot\frac{R_q}{2}+ \frac{1}{2}\cdot\frac{R_q}{2}= \frac{2R_p+7R_q}{12}.
\end{align*}

Now we have $\frac{2R_p+7R_q}{12}>\frac{R_p+R_q}{2} \iff 4R_p < R_q$: thus, if
$4R_p<R_q$, then $U^1_{\overline{P},P}>U^1_{P,P}$, and by symmetry,
$U^2_{P,\overline{P}}>U^2_{P,P}$. Further, since it is a constant sum game, the
latter inequality implies
$U^1_{P,\overline{P}}<U^1_{\overline{P},\overline{P}}$ (and
$U^2_{\overline{P},P}<U^2_{\overline{P},\overline{P}}$). Thus $\overline{P}$, a
non-PPS strategy, is a dominant strategy for both players On the other hand, if
$4R_p > R_q$, a similar argument shows that $PPS$ is a dominant strategy for
both players. 
\end{proof}

Thus, it is non-trivial to ensure PPS. On the other hand, the following example
suggests it is necessary for efficient research, by showing that the makespan in
the absence of PPS can be greater by a factor proportional to the number of
agents:

\begin{example}
\label{eg:poanopps}
Consider a treasure-hunt game under the SA model on a linear subtask network,
with $m\sim\Omega(\log n)$ subtasks, each with simplicity $1$, being attempted
by $n$ agents, each having ability $1$. Suppose all agents agree to follow PPS
-- then, the expected makespan is $\EE[T_{PPS}]=m/n.$

On the other hand, suppose the rewards are distributed so that all subtasks have
a reward of $0$, except the last subtask which has a reward of $1$ -- it is easy
to see that withholding partial-progress is a dominant-strategy for all agents
in this case. Further, the makespan is given by
$T_{Non-PPS}=\min_{i\in[n]}\left\{\sum_{u\in[m]}T_i(u)\right\}$, where
$T_i(u)\sim Exp(1)$. Since $m=\Omega(\log n)$, via a standard Chernoff bound, we
have that for all agents, $\sum_{u\in[m]}T_i(u)=\Omega(m)$ with high probability
-- thus we have that $\EE[T_{Non-PPS}]=\Omega(m)$.
\end{example}

Taken together, the examples raise the following two questions: Can one design
subtask rewards to ensure PPS under general aptitude matrices and acyclic
subgraphs? Is ensuring PPS sufficient to minimize makespan in different
settings? Our work provides answers to both these questions.

\subsection{Overview of our Results}
\label{ssec:overview}

For ease of exposition, we will summarize our results here using under the SA model. Extensions to the general aptitudes model will be made clearer in the following sections.

\noindent\textbf{The Treasure-Hunt Game on a Line:} We first focus on the
special case of a linear subtask-network (see figure \ref{fig:thunt}$(a)$). At
any time, each agent has only one available subtask; upon solving it, she is
free to choose when and with whom to share the solution. In this setting, we have the following result: 

\begin{result}
On linear subtask-networks under the SA model, suppose rewards $\{R_t\}$ satisfy
the following condition -- for every pair of subtasks $(u,v)$ s.t. $u$ precedes
$v$:
\begin{align*}
	\frac{R_us_u}{R_vs_v}\geq \alpha,
\end{align*}
where $\alpha=\max_{i\in[n]}\left\{a_i/\sum_{j\in[n]} a_j\right\}$. Then PPS is
a (subgame perfect) Nash Equilibrium.
\end{result}

\noindent Some comments and extensions (for details, refer Section
\ref{ssec:extensions}):
\begin{itemize}[nolistsep,noitemsep]
\item[$\bullet$] \emph{Sufficient condition for general aptitude matrix:} The
above result follows from more general conditions that ensure PPS under
\emph{any aptitude matrix} (Theorem \ref{thm:LineNE}). However, under the SA
model, the condition admits the following nice interpretation:
\item[$\bullet$] \emph{The Proportional-Allocation Mechanism:} Under the SA
model, for any agent, the ratio of expected completion times for any two tasks
is inversely proportional to their simplicities. Thus, a natural `locally-fair'
mechanism is to set subtask rewards to be inversely proportional to their
simplicity -- the above result shows that this surprisingly also ensures global
information sharing. 
\item[$\bullet$] \emph{Approximate Proportional-Allocation:} Moreover, the above
result shows that PPS is incentive compatible even under \emph{approximate}
proportional-allocation -- as long as rewards for earlier subtasks are at least
an $\alpha$ fraction of rewards for later subtasks. Note that $\alpha$
corresponds to the fraction of the total ability possessed by the best agent --
in particular, $\alpha<1$.
\item[$\bullet$] \emph{PPS payoffs in the core:} Our results extend to
transferable-utility cooperative settings, where agents can form groups -- in
particular, we get conditions for ensuring that the vector of payoffs from PPS is in the core (Theorem \ref{thm:LineCore}). Moreover,
the modified conditions admit the proportional-allocation mechanism under the SA
model (Corollary \ref{cor:SA_core}). 
\item[$\bullet$] \emph{Uniqueness of equilibrium with Stackelberg agents:}
Finally, we modify the above result to show that following PPS can be shown to
be the \emph{unique} Nash Equilibrium by admitting Stackelberg
strategies -- wherein a group of agents ex-ante commit to following PPS (Corollary \ref{cor:SA_DSIC}). In
particular, we show that a sufficient condition for uniqueness is that for any
task, the Stackelberg agents together dominate any other individual agent. For a generalization to the general aptitudes model, see Theorem \ref{thm:LineDSIC} in Section \ref{appendix}.
\end{itemize}

\noindent\textbf{The Treasure-Hunt Game on Acyclic Subgraph-Networks:} Analyzing
the treasure-hunt game is more complicated in general acyclic subtask-networks
where, at any time, an agent may have several available subtasks, and can split
her effort between them. In particular, given available subtasks $M$, an agent
$i$ can choose subtasks according to any distribution $\{x_i(u)\}_{u\in M}$.
Despite this added complexity, we obtain extensions to our
previous result for the line:

\begin{result}
On an acyclic subtask-network under the SA model, suppose rewards
$\{R_t\}_{t\in[m]}$ satisfy the following condition: for every pair of subtasks
$(u,v)$ such that $v$ is reachable from $u$ in the subtask-network, we have:
\begin{align*}
	R_us_u\geq R_vs_v.
\end{align*}
Then PPS is a (subgame perfect) Nash Equilibrium.
\end{result}

Some comments and extensions regarding this setting (see Section \ref{sec:DAG}
for details):
\begin{itemize}[nolistsep,noitemsep]
\item[$\bullet$] For \emph{general aptitudes matrix}, we characterize conditions under which
PPS is incentive-compatible (Theorem \ref{thm:DAG}).
\item[$\bullet$] Unlike linear subtask-networks, we can no longer admit
approximate proportional-allocation because earlier tasks need to be weighted
more -- however, exact proportional allocation still incentivizes PPS. 
\item[$\bullet$] We again show conditions under which the vector of payoffs from PPS is in the core (Theorems \ref{thm:DAGwithSA}, \ref{thm:DAGCore}), and
how to ensure that PPS is the unique equilibrium allowing for \emph{Stackelberg} agents (Theorem \ref{thm:DAGDSIC}).
\item[$\bullet$] Ensuring PPS indirectly results in all agents
working on the same subtask at any time. Such behavior, referred to as
\emph{herding} \cite{Strevens2013}, has been reported in empirical studies
\cite{boudreau2014cumulative}. 
\end{itemize}

\noindent\textbf{The Efficiency of PPS in the Treasure-Hunt Game:} Finally, we
explore the connections between PPS and minimum makespan. Although PPS
intuitively appears to be a pre-requisite for minimizing makespan, this may not
be true in general acyclic networks. To quantify the efficiency of PPS, we study
the ratio of the \emph{optimal makespan} to that under any reward-vector which
incentivizes PPS. Note that under conditions wherein PPS is the unique
equilibrium, this is a natural notion of \emph{price of anarchy} (POA) for the
Treasure-Hunt game.

In linear subtask-networks, it is not hard to see that under any
aptitude-matrix, \emph{ensuring PPS is necessary and sufficient for minimizing
the makespan}. For a general acyclic network however, we show that the ratio of
the makespan under PPS to the optimal can be unbounded. Significantly, however,
we show that the ratio scales \emph{only in the number of tasks, and not the
number of agents}. Furthermore, in the case of the SA model, we have the
following result:

\begin{result}
Consider the Treasure-hunt game on any acyclic subtask-network under the SA
model. Suppose rewards are chosen to ensure that all agents follow PPS. Then any
resulting equilibrium also minimizes the expected makespan.
\end{result}	 

Note that the above result does not reference details of how agents choose which
subtask to solve -- PPS alone is sufficient to minimize makespan. 

\noindent\textbf{Interpreting our model and results:} Collaboration/competition
in research involves many factors -- trying to incorporate all of these in a
model may make the strategy-space so complicated as to not allow any analysis.
On the other hand, with a stylized model like ours, there is a danger that some
phenomena arise due to modeling artifacts. We address this concern to some
degree in Section \ref{sec:discussion}, where we argue that our qualitative
results are not sensitive to our assumptions. 

More significantly, our model captures the main tension behind information
sharing in research -- partial-progress in a problem by one agent creates an
externality in the form of information asymmetry. To promote PPS, the reward for
sharing must compensate the agent for removing this externality. Our model lets
us quantify this `price of information' -- in particular, we show that
proportional-reward mechanisms indirectly help ensure partial-progress sharing.
We also show that PPS, while necessary, is also often sufficient, for efficiency
in large-scale research. 

Our work suggests that an effective way to conduct research on a big project is
to first break the project into many small subtasks, and then award credits for
subtasks proportional to their difficulty. The former recommendation matches the
evidence from the Polymath project \cite{cranshaw2011polymath,polymath}, which
is based largely on this idea of aggregating small contributions. The latter
suggests that in order to scale such platforms, it is important to have
mechanisms for rewarding contributions. The fact that that simple mechanisms
suffice to ensure PPS also points the way to potentially having these rewards
defined endogenously by the agents.

\noindent\textbf{Technical Novelty:} Our model, though simple, admits a rich
state space and agent strategies, which moreover affect the state transitions.
Thus, a priori, it is unclear how to characterize equilibria in such settings.
Our results follow from careful induction arguments on appropriate subgames. We
define a class of subgames corresponding to nested pairs of connected subgraphs
of the network. In general acyclic networks, in addition to PPS, we also show
that agents follow a herding strategy -- this turns out to be crucial for
characterizing equilibria. Finally, our efficiency result under the SA model is
somewhat surprising. In a related one-shot game where agents need to choose one
among a set of available tasks \cite{kleinberg}, it has been shown that under a similar
separable-aptitudes model, though there exist reward-structures achieving the
social optimum (in this case, maximizing the number of completed tasks),
finding these rewards is NP-hard. In contrast, we show that ensuring PPS alone
is sufficient to minimize makespan, and give simple sufficient conditions for the
same, thus circumventing the hardness of static settings under the SA model. 

\subsection{Related Work}
\label{ssec:relwork}

Understanding the strategic aspects of contests has had a long
history \cite{konrad2009strategy}. Recently, there has been a focus on
problems related to distribution of credit among researchers, such as: the impact of
authorship order on research \cite{ackerman2012research}; 
collaboration across multiple projects under local profit-sharing rules \cite{vojnovic};
Considering the effect of the priority rule on agents choosing
between open problems, it has been shown that redistributing rewards can make the system
more efficient \cite{kleinberg}. Our work shares much with these works, in particular, in the
idea of engineering rewards to achieve efficiency. However, these works consider only 
static/one-shot settings. Another line of work focuses on 
dynamics in R{\&}D settings \cite{taylor1995digging,harris1987racing}. R{\&}D contests are modeled
as a one-dimensional race. However, dynamics of information-sharing are ignored in these models.

The information-sharing aspect of our work shares similarities with models for
crowdsourcing and online knowledge-sharing forums. There has been some work on 
modeling crowdsourcing platforms as all-pay auctions \cite{chawla2012optimal}, where agents pay a cost to
participate -- this however is not the case in research, where agents derive
utility only by solving open problems in their chosen field. Models for online Q{\&}A forums,
 where rewards are given to the top answers to a question, have also been studied \cite{jain2012designing,ghosh2012implementing}. 
 In these models, each user has some knowledge ex-ante, and by submitting an answer later, can aggregate earlier answers. In
research however, knowledge does not exist ex ante, but is created while working
on problems, as is captured in our model. Finally, information dynamics are
studied by \cite{ghosh2013incentivizing} in the context of online education
forums -- though differing in topic, our work shares much of their modeling
philosophy.

There is also a growing body of empirical work studying open information in
collaborative research. The need for better incentive structures has been underscored in a recent empirical study on the Polymath projects \cite{cranshaw2011polymath}. The value of
open information-sharing in collaborative research has also been demonstrated, via experiments in Topcoder contests \cite{boudreau2014cumulative,boudreau2014disclosure}. They observe that, intermediate rewards, while increasing the efficiency
of problem solving, also leads to herding -- both these findings correspond
to what we observe in our models.

\noindent\textbf{Outline:} The remaining paper is organized as follows: In
Section \ref{sec:treasurehunt}, we analyze projects with linear
subtask-networks, extend these results to general acyclic subtask-networks in Section \ref{sec:DAG}.
Next, in Section \ref{sec:poa}, we characterize the efficiency, with respect to the makespan,
of PPS in different settings. We conclude with a discussion in
Section \ref{sec:discussion}. Some theorems mentioned in the main text are laid out in Section \ref{appendix}.

\section{The Treasure-Hunt Game on Linear Subtask-Networks}
\label{sec:treasurehunt}

\subsection{Extensive-Form Representation of the Treasure-Hunt Game}
\label{ssec:thgamemodel}

We denote $[n]=\{1,2,\ldots,n\}$. The treasure-hunt game on a line is defined as
follows:
\begin{itemize}[nolistsep,noitemsep]
\item[$\bullet$] \textbf{Task Structure:} A project is divided into a set $[m]$
of smaller subtasks, arranged in a linear subtask-network. Subtask $u$ becomes
available to an agent only when she knows solutions to all preceding subtasks;
the overall project is completed upon solving the final subtask. A set of $n$
non-identical selfish agents are committed to working on the project. 

\item[$\bullet$] \textbf{Rewards:} Subtask $u$ has associated reward $R_u\geq
0$, awarded to the first agent to publicly share its solution. We assume that
\emph{the rewards are fixed exogenously}. 

\item[$\bullet$]  \textbf{System Dynamics:} We assume that \emph{agents work on
subtasks independently}. Agent $i$, working alone, solves subtask $u$ after a
delay of $T_i(u)\sim Exp(a_i(u))$, where we refer to $a_i(u)$ as the aptitude of
agent $i$ for subtask $u$. For the special case of separable-aptitudes (see
Definition \ref{def:multiapt}), we have $a_i(u)=a_is_u$, where we refer to
$\{a_i\}_{i\in[n]}$ as the agent-ability vector, and $\{s_u\}_{u\in[m]}$ as the
subtask-simplicity vector. 

\item[$\bullet$] \textbf{Strategy Space:} The game proceeds in continuous time
-- for ease of notation, we suppress the dependence of quantities on time. If an
agent $i$ solves subtask $u$ at some time $t$, then she can share the solution
with any set of agents $A\subseteq [n]$ at any subsequent time -- if the
solution is shared with all other agents, we say it is \emph{publicly shared}.
When agent $i$ shares the solution to subtask $u$ with an agent $j$, then $j$
can start solving the subtask following $u$.

\item[$\bullet$] \textbf{Information Structure:} We assume all agents ex ante
know the aptitude-matrix $\{a_i(t)\}_{i\in[n],t\in[m]}$ -- later we show how
this can be relaxed. During the game, we assume agents only know their own
progress, and what is shared with them by others.
\end{itemize}

\subsection{Progress Sharing in the Treasure-Hunt Game}
\label{ssec:thgamesoln}

Our main result in this section is a sufficient condition on rewards to
incentivize PPS:

\begin{theorem}
\label{thm:LineNE}
Consider the Treasure-hunt game on a linear subtask-network $G$, and a general
agent-subtask aptitude matrix $\{a_i(u)\}$. Suppose the rewards-vector $\{R_u\}$
satisfies the following: for any agent $i$ and for any pair of subtasks $u,v$
such that $u$ precedes $v$, the rewards satisfy:
\begin{align*}
\frac{R_ua_{-i}(u)}{R_va_{-i}(v)}\geq \frac{a_i(v)}{a_i(v)+a_{-i}(v)}, 
\end{align*}
where for any task $w$, we define $a_{-i}(w)\triangleq\sum_{j\neq i}a_j(w)$.
Then for all agents, partial-progress sharing (PPS) is a (subgame perfect) Nash Equilibrium.
\end{theorem}

\begin{proof}
We number the subtasks from the end, with the last subtask being denoted as $1$,
and the first as $m$. Fix an agent $i\in[n]$, and assume every other agent
always follows the PPS policy. We define $\mathcal{G}_{k,l}^i, 0\leq k\leq l\leq
m,$ to be the subgame where agent $i$ starts at subtask $k$, and all other
agents start at subtask $l$. Note that agent $i$ starts either at par or ahead
of others -- this is because we assume that all other agents follow PPS. When
$k<l$, agent $i$ is ahead of the others at the start -- however she has not yet
revealed the solutions to her additional subtasks. In this setting, we claim the
following invariant: under the conditions specified in the theorem, PPS is a
best response for agent $i$ in \emph{every} subgame $\mathcal{G}_{k,l}^i,1\leq
k\leq l\leq m$. 

First, some preliminary definitions: given two subgames $\mathcal{G}_{k,l}^i$
and $\mathcal{G}_{p,q}^i$, we say that $\mathcal{G}_{k,l}^i$ is smaller than
$\mathcal{G}_{p,q}^i$ if $k<p$ OR $k=p$ and $l<q$. Also, assuming the invariant
is true, the expected reward earned by agent $i$ in subgame
$\mathcal{G}_{k,l}^i$ is given by:
\begin{align*}
R(\mathcal{G}_{k,l}^i)=\sum_{u=k+1}^{l}R_u+\sum_{u=1}^{k}R_u.\left(\frac{a_i(u)}
{a_i(u)+a_{-i}(u)}\right),	
\end{align*}
where $a_{-i}(u)=\sum_{j\in[n]\setminus\{i\}}a_j(u)$. When $k<l$, we
refer to the (non-empty) subtasks $\{l,l-1,\ldots,k+1\}$ as the \emph{captive
subtasks} of agent $i$. By definition, only $i$ knows the solutions to her
starting subtasks at the start of game $\mathcal{G}_{k,l}^i$. 

We prove the invariant by a joint induction on $k$ and $l$. For the base case,
consider the subgames $\mathcal{G}_{0,l}^i$, which correspond to a game where
$i$ knows solutions to all subtasks -- clearly following PPS (i.e., sharing the
solutions at $t=0$) is a best response. Now we split the remaining subgames into
two cases: $k<l$ and $k=l$. 

\noindent$\bullet$\emph{ Case $i.\,\,(\mathcal{G}_{k,l}^i, 1\leq k<l)$}: 
Assume game $\mathcal{G}_{k,l}^i$ starts at $t=0$, and the invariant holds for
all smaller games. If agent $i$ follows PPS, and publicly shares her captive
subtasks at $t=0$, then $\mathcal{G}_{k,l}^i$ reduces to the smaller subgame
$\mathcal{G}_{k,k}^i$. Suppose instead $i$ chooses to defect from PPS. Let
$\tau_1$ to be the first time when either agent $i$ solves subtask $k$ OR some
agent $j\neq i$ solves subtask $l$. Via a standard property of the minimum of
independent exponential random variables, we have
$\tau_1\sim\min\{Exp(a_i(k)),Exp(a_{-i}(l)\}\sim Exp(a(i,k,l))$, where we define
$a(i,k,l) = a_i(k)+a_{-i}(l)$.

Next, we observe that it is sufficient to consider deviations from PPS of the
following type: $i$ chooses a fixed delay $\tau>0$, and publicly shares the
solutions to captive subtasks at time $t=\min\{\tau_1,\tau\}$. This follows from
the following argument $(*)$: 
\begin{itemize}[nolistsep,noitemsep]
\item Suppose $\tau_1<\tau$, and assume some agent $j\neq i$ solves subtask $l$
before $i$ solves subtask $k$. Then agent $j$ immediately shares the solution to
subtask $l$ (as other agents follow PPS). This reduces the game to
$\mathcal{G}_{k,l-1}^i$, wherein by our induction hypothesis, $i$ immediately
shares her remaining captive subtasks, i.e., $\{l-1,l-2,\ldots,k+1\}$.
\item If $\tau_1<\tau$, but agent $i$ solves subtask $k$ before any other agent
solves subtask $l$, then the game reduces to $\mathcal{G}_{k-1,l}^i$, wherein
again by our induction hypothesis, $i$ immediately shares her captive subtasks,
i.e., $\{l,l-1,\ldots,k\}$.
\item Any randomized choice of $\tau$ can be written as a linear combination of
these policies.
\end{itemize}
\noindent To complete the induction step for $\mathcal{G}_{k,l}^i$, we need to
show that $\tau=0$ is the best response for $i$ in the game
$\mathcal{G}_{k,l}^i$. Let $U^{i}_{k,l}(\tau)$ be the expected reward of agent
$i$ if she chooses delay $\tau$. Also, for any subtask $u$, let
$a(u)=a_i(u)+a_{-i}(u)$. Then for $\tau=0$, we have:
\begin{align*}
U^{i}_{k,l}(0) =& \sum_{u=k+1}^{l}R_u +
\sum_{u=1}^{k}R_u.\left(\frac{a_i(u)}{a(u)}\right)= R_{l} + \Delta +
\frac{a_i(k)}{a(k)}.R_{k} +R(\mathcal{G}_{k-1,k-1}^i),
\end{align*}
where we define $\Delta=\sum_{j=k+1}^{l-1}R_{j}$. 
$R(\mathcal{G}_{k-1,k-1}^i)$ is the expected reward earned by agent $i$ in the
sub-game $\mathcal{G}_{k-1,k-1}^i$ assuming the induction hypothesis. 
Next we turn to the case of $\tau>0$. Suppose $\tau_1<\tau$: conditioned on
this, we have that $i$ solves subtask $k$ before any other agent solves subtask
$l$ with probability $\frac{a_i(k)}{a(i,k,l)}$ (from standard properties of the
exponential distribution). Now for any $\tau>0$, we can write:
\begin{align*}
U^{i}_{k,l}(\tau)=&\PP[\tau_1\geq\tau]\left(\sum_{u=k+1}^{l}R_u+R(\mathcal{G}_{k
,k}^i)\right)
+\PP[\tau_1<\tau]\left(\frac{a_i(k)}{a(i,k,l)}R(\mathcal{G}_{k-1,l}^i)+\frac{a_
{-i}(l)}{a(i,k,l)}R(\mathcal{G}_{k,l-1}^i)\right)
\end{align*}
This follows from the argument $(*)$ given above. Expanding the terms, we have:
\begin{align*}
U^{i}_{k,l}(\tau)&=\PP[\tau_1\geq\tau]\left(R_{l}+\Delta+\frac{a_i(k)}{a(k)}R_{k
}+R(\mathcal{G}_{k-1,k-1}^i)\right) \\
&+\PP[\tau_1<\tau]\frac{a_i(k)}{a(i,k,l)}\Bigg(R_{l}+\Delta+R_{k}+R(\mathcal{G}_
{k-1,k-1}^i)\Bigg)\\
&+\PP[\tau_1<\tau]\frac{a_{-i}(l)}{a(i,k,l)}\left(\Delta+\frac{a_i(k)}{a(k)}R_{k
}+R(\mathcal{G}_{k-1,k-1}^i)\right)
\end{align*}
Simplifying using $a(i,k,l)=a_i(k)+a_{-i}(l)$ and
$\PP[\tau_1\geq\tau]+\PP[\tau_1<\tau]=1$, we get:
\begin{align*}
U^{i}_{k,l}(\tau)=&
R_{l}\left(1-\PP[\tau_1<\tau]\frac{a_{-i}(l)}{a(i,k,l)}\right)+\Delta
+ R_{k}\left(\frac{a_i(k)}{a(k)}+\PP[\tau_1<\tau]\frac{a_{-i}(k)}{a(k)}.\frac{a_
{i}(k)}{a(i,k,l)}\right) +R(\mathcal{G}_{k-1,k-1}^i)
\end{align*}
Subtracting this from our expression for $U^{i}_{k,l}(0)$, we have:
\begin{align*}
U^{i}_{k,l}(0) - U^{i}_{k,l}(\tau) =\frac{\PP[\tau_1<\tau]}{a(i,k,l)}\left(
R_{l}a_{-i}(l)-R_{k}\frac{a_{-i}(k)a_i(k)}{a(k)}\right).
\end{align*}
Finally from the condition in Theorem \ref{thm:LineNE}, we have
$\frac{R_{l}a_{-i}(l)}{R_{k}a_{-i}(k)}\geq \frac{a_i(k)}{a(k)}$: substituting we
get $U^{i}_{k,l}(0) - U^{i}_{k,l}(\tau)>0\fall \tau>0$. Thus setting $\tau=0$
maximizes expected reward.

\noindent$\bullet$\emph{ Case $ii.\,\,(\mathcal{G}_{k,k}^i, k>1)$}: 
Assume the invariant holds for all smaller games. In this case, note that the
game reduces to either $\mathcal{G}_{k-1,k}^i$ (if $i$ solves subtask $k$) or to
$\mathcal{G}_{k-1,k-1}^i$ (if some agent $j\neq i$ solves subtask $l$). For both
these subgames, we have that PPS is a best response for $i$ via our induction
hypothesis. This completes the induction.
\end{proof}

\begin{theorem}\label{thm:LineCore}
If for any coalition $C \subset [n]$ and for any pair of subtasks $u,v$
such that $u$ precedes $v$, the rewards satisfy:
\begin{align*}
\frac{R_u a_{-C}(u)}{R_va_{-C}(v)}\geq \frac{a_C(v)}{a_C(v)+a_{-C}(v)}, 
\end{align*}
where for any task $w$, we define $a_{C}(w)\triangleq\sum_{j \in C}a_j(w)$ and $a_{-C}(w) \triangleq \sum_{j \notin C}a_j(w)$,
then the vector of payoffs from the PPS policy is in the core.
\end{theorem}
\begin{proof}
Since the conditions of Theorem \ref{thm:LineNE} are met, PPS is a Nash equilibrium. The utility of each agent $i$ in this equilibrium is clearly given by $u_i \triangleq \sum_{u \in [m]} \frac{a_i(u)}{a_i(u) + a_{-i}(u)}R_u$.

Observe that the combined utility gained by a coalition is equivalent to that of a single player with completion-time given by the rate parameter $a_S(.) = \sum_{j \in S} a_j(.)$. By the condition in the theorem, for any coalition $C \subset [n]$, PPS is a best response to everyone else's following PPS. And the payoff $v(C)$ obtained by $C$ is equal to 
\begin{align*}
v(C) = \sum_{u \in [m]}\frac{a_C(u)}{a_C(u)+a_{-C}(u)}R_u = \sum_{u \in [m]} \sum_{i \in C} \frac{a_i(u)}{a_i(u)+a_{-i}(u)}R_u = \sum_{i \in C}u_i
\end{align*}
\end{proof}

Note that the first scenario in Example \ref{eg:thunt} satisfies the condition of Theorem \ref{thm:LineNE}, while the second violates it. As we have just seen, the proof of this theorem follows from a careful backward-induction argument. Making fuller use of this idea, we can extend this theorem in many interesting ways, as we discuss next.

\subsection{The SA model, and stronger equilibrium concepts:}
\label{ssec:extensions}

\begin{itemize}[nolistsep,noitemsep]
\item[$\bullet$] \emph{Approximate Proportional-Allocation in the SA model:}
Applying the theorem to the SA model gives the following corollary (presented as
Main Result $1$ in Section \ref{ssec:overview}):

\begin{corollary}
\label{cor:SA}
Consider the Treasure-hunt game on linear subgraph-network $G$ under the SA
model. Suppose $A\triangleq\sum_{i\in[n]}a_i$, and $\alpha_{NE}\triangleq
\max_{i\in[n]}\frac{a_i}{A}$. Then the PPS policy is a Nash equilibrium for all
agents if the rewards-vector satisfies:
\begin{align*}
\frac{R_us_{u}}{R_vs_{v}}\geq \alpha_{NE}, \,\,\fall (u,v)\mbox{ s.t. $u$
precedes $v$ in $G$}
\end{align*}
\end{corollary}

Note that under the SA model, for any agent $i$, the ratio of expected
completion time for subtasks $u$ and $v$ obeys
$\frac{\EE[T_i(u)]}{\EE[T_i(v)}=\frac{s_v}{s_u}$. Thus we can reinterpret the
above condition in terms of an \emph{approximate proportional-allocation
mechanism}, which rewards subtasks proportional to the expected time they take
to solve. Observe also that $\alpha_{NE}$ is the fraction of the total ability
possessed by the top agent. In a sense, the
top agent is most likely to gain by deviating from PPS, and the rewards should be calibrated so as to prevent this. 

\item[$\bullet$] \emph{PPS payoffs in the Core of Cooperative Settings:} As seen in Theorem \ref{thm:LineCore}, we can derive a sufficient condition for the vector of payoffs from PPS to be in the core: to ensure that no coalition of agents, who share
partial-progress and rewards amongst themselves, benefits by deviating
from the PPS policy. We state this result for the SA model as follows:

\begin{corollary}
\label{cor:SA_core}
Consider the Treasure-hunt game on linear subtask-network $G$, under the
separable-aptitudes model. Define $A \triangleq\sum_{i\in[n]}a_i$ and
$\alpha_{C}\triangleq 1-\min_{i\in[n]}\frac{a_i}{A}$. Then the vector of payoffs from the PPS policy is in
the core if the rewards-vector satisfies:
\begin{align*}
\frac{R_us_{u}}{R_vs_{v}}\geq \alpha_{C}, \,\,\fall (u,v)\mbox{ s.t. $u$
precedes $v$ in $G$}
\end{align*}
\end{corollary}

The corollary follows from Theorem \ref{thm:LineCore} by observing that $A \alpha_C = 1 - \min_{i \in [n]}a_i$ corresponds to the combined ``ability" of the most powerful coalition consisting of everyone except a single agent. Note that $\alpha_{NE}\leq \alpha_{C}<1$ -- thus the vector of payoffs from an approximate proportional-allocation mechanism is in the core, but with a stricter approximation factor. 

\item[$\bullet$] \emph{Stackelberg strategies:} 
One subtle point regarding Theorem \ref{thm:LineNE} is that the resulting
equilibrium need not be unique. For two agents, an argument similar to Example
\ref{eg:thunt} can be used to show uniqueness -- this approach does not easily
extend to more players. However, we can circumvent this difficulty and find
conditions under which PPS is the unique Nash equilibrium, by
considering settings with \emph{Stackelberg leaders} -- agents who \emph{ex ante
commit to following the PPS policy}. We state this here for the SA model:

\begin{corollary}
\label{cor:SA_DSIC}
Consider the Treasure-hunt game on linear subtask-network $G$ under the SA
model, with $A \triangleq\sum_{i\in[n]}a_i$. Suppose a subset of agents
$S\subseteq[n]$ ex ante commit to PPS. Define $\alpha_{S}=\max_{i\in [n]}
\frac{a_i}{a_S}\frac{a_{-i}}{a_{i}+a_{-i}}$, where
$a_S = \sum_{i \in S}a_i$. Then PPS is the unique Nash Equilibrium
if the rewards satisfy:
\begin{align*}
\frac{R_us_{u}}{R_vs_{v}}\geq \alpha_{S}, \,\,\fall (u,v)\mbox{ s.t. $u$
precedes $v$ in $G$}
\end{align*}
\end{corollary}

This follows from a more general result on general aptitudes, stated as Theorem \ref{thm:LineDSIC} in Section \ref{appendix}).
Note that if $a_S\geq a_i$, then $\alpha_S<1$ -- thus if the
Stackelberg agents dominate any individual agent, then proportional allocation
results in PPS. Such behavior mirrors real-life settings, wherein when top
researchers commit to open information sharing, then many others follow suit. An
example of this can be seen in the Polymath project, where a core group of
mathematicians first committed to open discussion, and this led to others
joining the project. 

\item[$\bullet$] \emph{Weaker Requirements on Ex Ante Knowledge:} We assumed
that all agents are aware of the entire aptitude matrix. Note however that the conditions in Theorem
\ref{thm:LineNE} depend only on an agent's own aptitude vector, and the
\emph{aggregate} aptitudes of the remaining agents. In particular, in the SA
model, an agent $i$ needs to only know $\alpha_i=a_i/A$, the ratio of her
ability to the total ability of all agents, in order to determine if PPS is a
best response or not. 

On the other hand, to correctly design exogenous rewards, a social planner needs
to know the simplicities of all subtasks. In case these estimates are not
perfect, but are within a $1\pm\epsilon$ factor of the true values, then as long
as $\epsilon>(1-\alpha_{NE})/(1+\alpha_{NE})$, then proportional-allocation
using the noisy estimates still incentivizes PPS.
\end{itemize}

\section{The Treasure-Hunt game on general directed acyclic graphs}
\label{sec:DAG}

In this section, we extend the results from the previous section to general acyclic subtask-networks. The subtask-network now encodes the following dependency constraints: a subtask $u$ with source node $v_u$ can be attempted by an agent $i$ only after $i$ knows the solutions to \emph{all} subtasks terminating at $v_u$ (i.e., all the in-edges of node $v_u$). Note that at any time, an agent may have multiple available subtasks. We assume that agent $i$ chooses from amongst her available subtasks $M$ according to some strategically-chosen distribution $\{x_i(u)\}_{u\in M}$ -- as a result, subtask $u\in M$ is solved by $i$ in time $Exp(x_i(u)a_i(u))$. The remaining model is as described in Section \ref{ssec:thgamemodel}. 

A technical difficulty in  acyclic subtask-networks is that even assuming agents
follow PPS, there is no closed-form expression for agent's utility in a subgame
-- consequently, at any intermediate stage, there is no simple way to determine
an agent's choice of available subtask. In spite of this, we can still derive a
sufficient condition on rewards to ensure PPS -- this also gives a constructive
proof of existence of reward-vectors that ensure PPS in any acyclic
subtask-network. 

We now introduce the following notion of a virtual reward:

\begin{definition}[\textbf{Virtual rewards}]
\label{def:and_gamma_lambda}
For any agent $i \subseteq [n]$ and subtask $u\in [m]$, the virtual-reward
$\gamma_i(u)$ is defined as: 
\begin{align*}
\gamma_i(u) &\triangleq \frac{R_u a_{i}(u) a_{-i}(u)}{a_{i}(u) + a_{-i}(u)},
\end{align*}
where $a_{-i}(u)=\sum_{j\in S\setminus\{i\}}a_j(u)$.
\end{definition}

Using this definition, we have the following theorem:

\begin{theorem}
\label{thm:DAG}
Consider the Treasure-hunt game on a directed acyclic subtask-network $G$, and a
general agent-subtask aptitude matrix $\{a_i(u)\}$. Suppose the rewards-vector
$\{R_u\}$ satisfies the following conditions:
\begin{itemize}[nolistsep,noitemsep]
 \item[$\bullet$] (Monotonicity) There exists a total ordering $\prec$ on
subtasks $[m]$ such that for every agent $i$ and pair of subtasks $u,v$, we have
$v\prec u\iff \gamma_i(v)<\gamma_i(u)$. 
 \item[$\bullet$] For all pairs of subtasks $(u,v)$ such that $v$ is reachable
from $u$ in $G$, we have $v\prec u$.
\end{itemize}
Then the following strategies together constitute a Nash
equilibrium:
\begin{itemize}[nolistsep,noitemsep]
\item[$\bullet$] Every agent implements the \emph{PPS} policy.
\item[$\bullet$] At any time, if $M_o\subseteq [m]$ is the set of available
subtasks, then every agent $i$ chooses to work on the unique subtask
$u^*=\argmax_{u\in M_o}\gamma_i(u)$
\end{itemize}
\end{theorem}

The proof of Theorem \ref{thm:DAG}, which is provided below, is again based on a careful backward induction argument. The
induction argument, however, is more involved than in Theorem \ref{thm:LineNE}
-- in particular, when multiple subtasks are available to each agent, the
utility of the subgame for an agent depends on the strategies of all other
agents. We circumvent this by observing that under the given conditions, besides
following PPS, \emph{agents also have an incentive to choose the same subtask as
other agents}. We henceforth use the term \emph{herding} to refer to this
behavior. 

\begin{proof}[Proof of Theorem \ref{thm:DAG}]

In an acyclic subtask-network, a subtask $v$ can be solved only if an agent knows the solutions to all
subtasks corresponding to edges from which $v$ is reachable -- the set $L(v)$ of predecessors of $v$. We define the set of
\emph{valid knowledge-subgraphs} $\mathcal{V} \triangleq \{ P \subseteq [m]| v \in P \implies L(v) \subset P\}$ -- subtasks whose solutions may be jointly
known to agents at some time.

Finally, given that an agent knows the solution to subtasks $P\in\mathcal{V}$,
the complement $P^c=[m]\setminus P$ represents the subtasks still open for the
agent. We define the set $\mathcal{S}=\{T\subseteq
M\subseteq[m]:M^c,T^c\in\mathcal{V}\}$ -- essentially $\mathcal{S}$ contains all
nested-pairs of subtask-sets, whose complements are valid knowledge-subgraphs. 
For example, for a linear subtask-network, where subtasks $[m]$ are arranged in
increasing order, for any $l\leq k$, the pair of subgraphs
$\{l+1,\ldots,m\}\subseteq\{k+1,\ldots,m\}$ are elements of $\mathcal{S}$.
Recall that these were the exact subgraphs we used for induction in Theorem
\ref{thm:LineNE} -- we now do something similar over pairs of subgraphs in
$\mathcal{S}$.

Fixing agent $i$, we define $\G_{M,T}^i, (M,T)\in\mathcal{S}$ to be the subgame
where agent $i$ knows the solutions to subtasks $T^c$, while any agent $j\neq i$
knows the solutions to subtasks $M^c$. Note that $T\subseteq M$, and
$T^c\setminus M^c$ are the captive subtasks of $i$, i.e., those which $i$ has
solved, but not publicly shared. Also, for any set $M$ such that
$M^c\in\mathcal{C}$, we define $M_o\subseteq M$ to be the available subtasks in
$M$.

Given this set of subgames, we claim the following invariant: under the
conditions specified in the theorem, PPS $+$ herding is a best response for
agent $i$ in \emph{every} subgame $\G_{M,T}^i$, $(M,T)\in S$. As in Theorem
\ref{thm:LineNE}, we prove this via induction. We first partition the
set $\mathcal{S}$ as follows: given $1\leq k\leq l\leq m$
\begin{align*}
 \mathcal{S}_{l,k} \triangleq \{ (M,T)\in\mathcal{S}:|M|=l,|T|=k\}
\end{align*}
We abuse notation and say subgame $\G_{M,T}^i\in\mathcal{S}_{l,k}$ if
$(M,T)\in\mathcal{S}_{l,k}$ -- note that this corresponds to agent $i$ having
$k$ unsolved subtasks, while the remaining agents have $l\geq k$ unsolved
subtasks. We can restate the above invariant as follows -- for all $1\leq k\leq
l\leq m$, and for all games $\G_{M,T}^i\in\mathcal{S}_{l,k}$, PPS $+$ herding is
a best response for agent $i$. 

For the base case, consider the set of subgames $\G_{M,T}^i\in
\mathcal{S}_{1,1}$, i.e., subgames where all agents have a single open subtask.
Clearly PPS $+$ herding is a best response for $i$ in this case. For the
inductive step, fixing $(l,k)$, we assume the invariant is true for all smaller subgames
in $\mathcal{S}_{q,p}$ where $q<l$ and $p\leq k$ (i.e., agent $i$ has strictly
less open subtasks), or, $q=l$ and $p<k$ (i.e., agent $i$ has $k$ open subtasks,
but the remaining agents have strictly less open subtasks). 

Now, as in Theorem \ref{thm:LineNE}, we break the induction into two cases:
$(i)\,k<l$, and $(ii)\,k=l$. Note that $k<l$ corresponds to subgames
$\G_{M,T}^i$ with $T\subset M$ -- here we need to show that PPS, i.e., $i$
sharing solutions to captive subtasks $T^c\setminus M^c$ at $t=0$, is a best
response. On the other hand, in any subgame $\G_{M,M}^i\in\mathcal{S}_{k,k}$, we
need to show herding, i.e., $i$ chooses to work on subtask $v^*=\argmax_{u\in
M_o}\gamma_i(u)$. We note that we do not need to assume/show that $i$ follows
herding in subgame $\G_{M,T}^i$ -- we show that PPS is a best response
regardless of which subtask in $T_o$ that $i$ attempts.

\noindent$\bullet$ Case (i) $\,\,(\mathcal{S}_{l,k}, 1\leq k<l)$\\
Consider subgame $\G_{M,T}^i \in \mathcal{S}_{l,k}$. As in Theorem
\ref{thm:LineNE}, the only deviation strategies consist of agent $i$ choosing a
delay parameter $\tau>0$, and sharing her captive subtasks at time
$\min\{\tau,\tau_1\}$, where $\tau_1$ is the time of the first breakthrough in a
subtask, either by her, or some agent $j\neq i$. Suppose agent $i$ concentrates
her efforts on some subtask $v\in T_o$; further, from our induction hypothesis,
we know any agent $j\neq i$ concentrates on subtask $u^*=\argmax_{u\in
M_o}\gamma_j(u)$ -- note that this is the same subtask $\fall j\neq i$.

If $\tau=0$ or $\tau<\tau_1$, then the game reduces to $\G_{T,T}^i$. If
$\tau_1<\tau$ and agent $i$ solves subtask $v$ first, then by our induction
assumption, agent $i$ immediately shares her remaining subtasks and the game
reduces to $\G_{T-v,T-v}^i$, where we use the notation $T-v=T\setminus\{v\}$.
Finally, if $\tau_1<\tau$ and some agent $j\neq i$ solves subtask $u^*$, then
again agent $i$ follows PPS and the game reduces to $\G_{T-u^*,T-u^*}^i$ -- note
though that task $u^*$ could be in $T^c$ (i.e., its solution be known to agent
$i$). For the remaining proof, \emph{the only deviation of $i$ from PPS we
consider is $\tau=\infty$}. To see why this is sufficient, note that for any
$\tau$, in sample paths where $\tau<\tau_1$, agent $i$ shares all her captive
subsets (albeit with a delay) before any agent makes an additional breakthrough.
The memoryless property of exponential random variables implies that choosing
any $\tau<\infty$ only scales the difference in utility with and without PPS, by
a factor of 
$\PP[\tau_1<\tau]$\footnote{This fact is somewhat non-trivial, but can be
formally derived, following the same steps as in the proof of Theorem
\ref{thm:LineNE} -- we omit it for the sake of brevity.}.

As before, for subtask $u$ we define $a(u)=\sum_j a_j(u)$, and for any pair of
subtasks $u,w$, we define $a(i,u,w)= a_i(u)+a_{-i}(w)$. Now for any $\G_{M,T}^i
\in \mathcal{S}_{l,k}$, defining $U_{M,T}^i$ to be the expected reward earned by
agent $i$ by following PPS, we have:
\begin{align*}
 U^{i}_{M,T} = \sum_{u\in T^c\setminus M^c}R_u + \sum_{u\in T}
\frac{a_i(u)}{a(u)} R_u.
\end{align*}
Note also that by our induction hypothesis, this is true for any pair
$(P,Q)\in\mathcal{S}$ such that $P\subseteq M,Q\subset T$ or $P\subset M,Q=T$.

Let $U_{M,T}^{i,dev}$ to be the expected payoff of agent $i$ when she deviates
from PPS by waiting till $\tau_1$ before sharing solutions. Now if $u^*\in T^c$,
we have: 
$$U_{M,T}^{i,dev}=\frac{a_{-i}(u^*)}{a(i,u^*,v)}U_{M-u^*,T}^i +
\frac{a_i(v)}{a(i,u^*,v)}U_{M,T-v}^i$$
Subtracting from the expected utility under PPS, we get: 
\begin{align*}
 U_{M,T}^i - U_{M,T}^{i,dev} &= \frac{a_{-i}(u^*)}{a(i,u^*,v)}\left(U_{M,T}^i -
U_{M-u^*,T}^i \right)+ \frac{a_i(v)}{a(i,u^*,v)}\left(U_{M,T}^i -
U_{M,T-v}^i\right). \\
&= \frac{a_{-i}(u^*)}{a(i,u^*,v)}\left(R_{u^*} \right)-
\frac{a_i(v)}{a(i,u^*,v)}\left( \frac{a_{-i}(v)}{a(v)}R_v \right)\\&
\geq \frac{a_{-i}(u^*)}{a(i,u^*,v)}\left(\frac{a_i(u^*)}{a(u^*)}
R_{u^*} \right)- \frac{a_i(v)}{a(i,u^*,v)}\left( \frac{a_{-i}(v)}{a(v)}R_v
\right)=\frac{\gamma_i(u^*) - \gamma_i(v)}{a(i,u^*,v)}.
\end{align*}
\noindent On the other hand, if $u^* \notin T^c$, we have: 
\begin{align*}
U_{M,T}^{i,dev}&=\frac{a_{-i}(u^*)}{a(i,u^*,v)}U_{M-u^*,T-u^*}^i+\frac{a_i(v)}{
a(i,u^*,v)}U_{M,T-v}^i. 
\end{align*}
Again subtracting from $U_{M,T}^i$, we get:
\begin{align*}
U_{M,T}^i - U_{M,T}^{i,dev} &= \frac{a_{-i}(u^*)}{a(i,u^*,v)}\left(U_{M,T}^i
-U_{M-u^*,T-u^*}^i\right)+\frac{a_i(v)}{a(i,u^*,v)}\left(U_{M,T}^i-U_{M,T-v}
^i\right),\\
&=\frac{a_{-i}(u^*)}{a(i,u^*,v)}\left(\frac{a_i(u^*)}{a(u^*)}R_{u^*}
\right)-\frac{a_i(v)}{a(i,u^*,v)}\left( \frac{a_{-i}(v)}{a(v)}R_v
\right)=\frac{\gamma_i(u^*) - \gamma_i(v)}{a(i,u^*,v)}.
\end{align*}
From the induction hypothesis, we have that for all agents $j\neq i$, herding
implies that $\gamma_j(u^*)>\gamma_j(w)\fall w\in M_o$. Further, from the
monotonicity condition (i.e., all agents have the same ordering on virtual
rewards), we have that $\gamma_i(u^*)>\gamma_i(w)\fall w\in M_o$. On the other
hand, for any task $w\in M_o$ such that $v$ is reachable $w$, then from the
second condition in the theorem, we have that $\gamma_i(w)\geq\gamma_i(v)$.
Putting these together, we get that $\gamma_i(u^*)\geq\gamma_i(v)$; plugging
into the above inequalities, we get that $U_{M,T}^i - U_{M,T}^{i,dev}\geq 0$.
Thus following PPS is a best response for agent $i$.

\noindent$\bullet$ Case (ii). $\,\,(\mathcal{S}_{k,k}, k>1)$\\
Consider a subgame $\mathcal{G}_{M,M}^i\in\mathcal{S}_{k,k}$. In this case, we
need to show that $i$ follows herding, i.e., $i$ chooses available subtask
$u^*=\argmax_{w\in M_o}\gamma_i(w)$. Let $ U^{i}_{M,M}$ be the expected reward
earned by $i$ if she attempts $u^*$, and  $U^{i,dev}_{M,M}$ be her expected
reward if she attempts some other subtask $v\in M_o$. Note that due to our
induction hypothesis, all other agents concentrate on $u^*$, and share it once
they solve it. Now we have: 
\begin{align*}
U_{M,M}^i - U_{M,M}^{i,dev} &= \frac{a_{-i}(u^*)}{a(i,u^*,v)}\left(U_{M,M}^i
-U_{M-u^*,M-u^*}^i\right)+\frac{a_i(v)}{a(i,u^*,v)}\left(U_{M,M}^i-U_{M,M-v}
^i\right),\\
&=\frac{a_{-i}(u^*)}{a(i,u^*,v)}\left(\frac{a_i(u^*)}{a(u^*)}R_{u^*}
\right)-\frac{a_i(v)}{a(i,u^*,v)}\left( \frac{a_{-i}(v)}{a(v)}R_v
\right)=\frac{\gamma_i(u^*) - \gamma_i(v)}{a(i,u^*,v)}.
\end{align*}
Since $u^*=\argmax_{w\in M_o}\gamma_i(w)$, therefore herding is a best response
for agent $i$. This completes the proof.
\end{proof}

Theorem \ref{thm:DAG} can be summarized as stating that as long as virtual
rewards are monotone and have the same partial-ordering as induced by the
subtask-network, then \emph{PPS $+$ herding is a Nash equilibrium}. As with the
condition in Theorem \ref{thm:LineNE}, this result may be difficult to
interpret. Note however that given \emph{any acyclic subtask-network, and any
aptitudes matrix} with non-zero aptitudes, it gives a constructive way to choose
$\{R_u\}$ to ensure PPS. The observation that ensuring PPS may also result in
herding is an interesting outcome of our analysis -- such behavior has in fact
been observed in research settings, for example, 
such behavior in contests with incentives for partial-progress sharing is reported in \cite{boudreau2014cumulative}.

In the SA model, note that for every agent $i$, the virtual reward
$\gamma_i(u)\propto R_us_u$. Thus for \emph{any} rewards that are inversely proportional to the ``simplicities", the
monotonicity property is satisfied. Moreover, we obtain the following sharper
characterization of sufficient conditions for ensuring PPS:

\begin{theorem}
\label{thm:DAGwithSA}
Consider the treasure-hunt game under the SA model, on a directed acyclic
subtask-network with $m$ subtasks, associated simplicities 
$\{s_k\}_{k\in\mathcal{M}}$ and exogenous (non-zero) rewards 
$\{R_k\}_{k\in[m]}$, and $n$ agents with associated abilities
$\{a_i\}_{i\in[n]}$. If the rewards satisfy: 
\begin{itemize}[nolistsep,noitemsep]
 \item[$\bullet$] $v \mathrm{~is~reachable~from~} u \implies R_vs_v \leq
R_us_u$.
\end{itemize}
Then any set of strategies jointly constitute a Nash Equilibrium if:
\begin{itemize}[nolistsep,noitemsep]
\item[$\bullet$] Every agent follows the \emph{PPS} policy.
\item[$\bullet$] At any time, if $M_o$ denotes the set of open subtasks, then
for any agent $i$, the support of her choice distribution $\{x_i(u)\}_{u\in
M_o}$ is contained in  $\argmax_{u\in M_o}\{R_u s_u\}$.
\end{itemize}
Moreover, the resulting payoffs-vector is in the core.
\end{theorem}

\begin{proof}
Consider a valid subgame $\G_{M,T}^i$ (as discussed in Theorem \ref{thm:DAG}).
Assume everyone but agent $i$ follows PPS $+$ herding -- given a set $M_o$ of
available subtasks, this means that for agent $j\neq i$, the support of
$\{x_j(u)\}_{u \in \M}$ is a subset of $S=\argmax_{u \in \M_o}\{R_us_u\}$. 
 
Suppose agent $i$ deviates by extending the support of $\{x_j(u)\}_{u \in \M}$
to a set $V$ of subtasks in $M_o\setminus S$. Define $\lambda_u = \sum_{j \in
\N/\{i\}} x_j(u)a_j$ to be the cumulative rate contributed by everyone except
agent $i$ on each $u \in S$, and $r_t = x_i(t)a_i$ to be the rate contributed by
agent $i$ on each $t \in S \cup V$ -- from definition, we have:
\begin{align*}
\sum_{u \in S} \lambda_u =  \sum_{u \in S} \sum_{j\neq i} x_j(u)a_j= \sum_{j\neq
i} a_j = a_{-i},\,
\sum_{u \in S} r_u + \sum_{v \in V}r_v=\sum_{u \in S \cup V}x_i(u)a_i=a_i
\end{align*}
For brevity, we define: $\alpha\triangleq \sum_{u \in S}s_u(r_u + \lambda_u) +
\sum_{v \in V} s_v r_v,\,c \triangleq R_us_u , ~\fall u \in S$. Now note that in
any game $\mathcal{G}_{M,T}^i$ agent $i$'s utility $U^i_{M,T}$ under PPS $+$
herding is given by:
\begin{align*}
U^i_{M,T} = \underset{t \in M/T}{\sum} R_t + U^i_{T,T}.
\end{align*}

We first show that the value $U^i_{T,T}$, after allowing for mixed strategies
with support in $S$, is the same as the value obtained in a situation where all
agents focus on the same task, i.e., $U^i_{T,T} = \underset{t \in T}{\sum}
\frac{a_t}{a}R_t$, where $a \triangleq \sum_{j \in \N} a_j$. This is definitely
true for all subgames in $S_{1,1}$. Assume it is true for all subgames in
$S_{k,k}$, where $1 \leq k < |T|$. Using this, we get, from standard properties
of exponential random variables:
\begin{align*}
U^i_{T,T} &= \underset{u \in S}{\sum} \frac{s_u r_u}{\alpha}R_u + \underset{u
\in S}{\sum} \frac{s_u(r_u + \lambda_u)}{\alpha}U^i_{T-u,T-u}  \\
&= \underset{u \in S}{\sum} \frac{r_u c}{\alpha}+\underset{u \in S}{\sum}
\frac{s_u(r_u+\lambda_u)}{\alpha} \left( \underset{t \in T-u}{\sum}
\frac{a_i}{a}R_t \right)\\
&= \underset{u \in S}{\sum} \frac{r_u c}{\alpha}+\underset{u \in S}{\sum}
\frac{s_u(r_u + \lambda_u)}{\alpha} \left( -\frac{a_u}{a}R_u + \underset{t \in
T}{\sum} \frac{a_i}{a}R_t \right) \\
  &= \underset{t \in T}{\sum} \frac{a_i}{a}R_t + \underset{u \in S}{\sum}
\frac{r_u c}{\alpha}-\underset{u \in S}{\sum} \frac{c(r_u + \lambda_u)}{\alpha}
\left( \frac{a_i}{a}  \right) \\
&= \underset{t \in T}{\sum} \frac{a_i}{a}R_t + \underset{u \in S}{\sum}
\frac{r_u c}{\alpha} \left( 1 - \frac{a_i}{a} \right)-\underset{u \in S}{\sum}
\frac{c\lambda_u}{\alpha} \left( \frac{a_i}{a}\right) \\
  &= \underset{t \in T}{\sum} \frac{a_i}{a}R_t + \underset{u \in S}{\sum}
\frac{r_u c}{\alpha} \left( \frac{a_{-i}}{a} \right)-\underset{u \in S}{\sum}
\frac{c\lambda_u}{\alpha} \left( \frac{a_i}{a}  \right) \\
&= \underset{t \in T}{\sum} \frac{a_i}{a}R_t,
\end{align*} 
Consequently, assuming $i$ follows PPS $+$ herding in $\mathcal{G}_{M,T}^i$, we
have:
\begin{align*}
U^i_{M,T} = \underset{t \in M/T}{\sum} R_t + \underset{t \in T}{\sum}
\frac{a_i}{a}R_t
\end{align*}

Having characterized the $i$'s value of any subgame $\G_{M,T}$ under PPS $+$
herding, we now need to show that any deviation from it reduces agent $i$'s
utility. This requires a more careful analysis as compared to Theorem
\ref{thm:DAG}, as we need to account for all mixed strategies. Further, some of
the subtasks in $S$ could belong to the set $M\setminus T$ of $i$'s
\emph{captive subtasks}. For this, we define: $S_i = S \cap T^c, S_{-i} = S \cap
T$. Clearly agent $i$ does not choose any subtask in $S_i$. Thus, the game
$\G_{M,T}^i$ transitions to smaller subgames as follows:
\begin{itemize}
\item $\fall u \in S_{-i}$ : $\G_{M,T}^i\rightarrow\G_{M,T-u}$ with probability
$\frac{s_u r_u}{\alpha} $, $\G_{M-u,T-u}$ with probability $\frac{s_u
\lambda_u}{\alpha}$. 
\item $\fall u^\prime \in S_i$ : $\G_{M,T}^i\rightarrow\G_{M-u^\prime,T}$ with
probability $\frac{s_{u^\prime} \lambda_{u^\prime}}{\alpha}$.
\item $\fall v \in V$ : $\G_{M,T}^i\rightarrow\G_{M,T-v}$ with probability $
\frac{s_v r_v}{\alpha}$.
\end{itemize}
From the expression we derive above for $U^i_{M,T}$, we have: $U^i_{M,T-u}=
U^i_{M,T} + R_u - \frac{a_i}{a}R_u,\, U^i_{M-u,T-u}= U^i_{M,T} -
\frac{a_i}{a}R_u,\,U^i_{M-u^\prime,T}=U^i_{M,T} - R_{u^\prime},\,U^i_{M,T-v}=
U^i_{M,T} + R_v - \frac{a_i}{a}R_v$. Now we can characterize the utility
$U^{i,dev}_{M,T}$ of agent $i$ in the case she deviates using the inductive
assumption that in every smaller subgame, agent $i$ follows PPS $+$ herding. We
present only the main calculations here, as the reasoning follows the same lines
as in the proof of Theorem
\ref{thm:DAG}:
\begin{align*}
 U^{i,dev}_{M,T} &= \underset{u \in S_{-i}}{\sum} \frac{s_u (r_u +
\lambda_u)}{\alpha}\left( \frac{r_u}{r_u + \lambda_u}U^i_{M,T-u} +
\frac{\lambda_u}{r_u + \lambda_u}U^i_{M-u,T-u}\right)\\
  &+ \underset{u^\prime \in S_{i}}{\sum} \frac{s_u \lambda_u}{\alpha}U^i_{M-u,T}
  + \underset{v \in V}{\sum} \frac{s_v r_v}{\alpha}U^i_{M,T-v} \\
   &= \underset{u \in S_{-i}}{\sum} \frac{s_u (r_u + \lambda_u)}{\alpha}\left(
\frac{r_u}{r_u + \lambda_u}\left(U^i_{M,T} + R_u - \frac{a_i}{a}R_u \right) 
   + \frac{\lambda_u}{r_u + \lambda_u} \left( U^i_{M,T} - \frac{a_i}{a}R_u
\right) \right)\\
  &+ \underset{u^\prime \in S_{i}}{\sum} \frac{s_u
\lambda_u}{\alpha}\left(U^i_{M,T} - R_u \right)
  + \underset{v \in V}{\sum} \frac{s_v r_v}{\alpha}\left(U^i_{M,T} + R_v -
\frac{a_i}{a}R_v \right)
\end{align*}
\noindent To complete the induction step, we need to show that
$U^{i,dev}_{M,T}\leq U^{i}_{M,T}$. We have:
\begin{align*}
  U^{i,dev}_{M,T} - U^{i}_{M,T}&= \underset{u \in S_{-i}}{\sum} \frac{s_u (r_u +
\lambda_u)}{\alpha}\left( \frac{r_u}{r_u + \lambda_u}\left(R_u -
\frac{a_i}{a}R_u \right) 
   + \frac{\lambda_u}{r_u + \lambda_u} \left(- \frac{a_i}{a}R_u \right)
\right)\\
  &+ \underset{u^\prime \in S_{i}}{\sum} \frac{s_u \lambda_u}{\alpha}\left(- R_u
\right)
  + \underset{v \in V}{\sum} \frac{s_v r_v}{\alpha}\left(R_v - \frac{a_i}{a}R_v
\right) \\
  &= \underset{u \in S_{-i}}{\sum} \frac{s_u r_u}{\alpha} \frac{a_{-i}}{a}R_u 
  - \underset{u \in S_{-i}}{\sum} \frac{s_u \lambda_u}{\alpha}
\frac{a_{i}}{a}R_u 
  - \underset{u^\prime \in S_{i}}{\sum} \frac{s_u \lambda_u}{\alpha} R_u 
  + \underset{u \in V}{\sum} \frac{s_v r_v}{\alpha} \frac{a_{-i}}{a}R_v 
 \end{align*} 

 Consider a subtask $v \in V$, one that agent $i$ devotes attention to outside
of $S$. One of the following conditions
 has to be satisfied,
 \begin{itemize}
  \item $v \in M_o/S$,
  \item $\exists t \in M_o$ such that $v$ is reachable from $t$
 \end{itemize}
In the former case, because of how we defined $S$, we have $R_vs_v < c$. In the
latter case, we have $R_ts_t < c$. And,
from the condition imposed in the hypothesis of the theorem, the fact that $v$
is reachable from $t$ implies that
$R_vs_v \leq R_ts_t < c$. Using this we can show that agent $i$ does not benefit
from deviation:
 \begin{align*}
  U^{i,dev}_{M,T} - U^{i}_{M,T} 
&< \frac{a_{-i}}{a\alpha} \left( \underset{u \in S_{-i}}{\sum} R_u s_u r_u 
  + \underset{v \in V}{\sum} R_v s_v r_v  \right)
  - \frac{a_i}{a\alpha} \left( \underset{u \in S_i}{\sum} R_u s_u \lambda_u 
  + \underset{u \in S_{-i}}{\sum} R_u s_u \lambda_u \right) \\
&< \frac{c a_{-i}}{a\alpha} \left( \underset{u \in S_i}{\sum} r_u 
  + \underset{v \in V}{\sum} r_v  \right)
  - \frac{c a_i}{a\alpha} \left( \underset{u \in S_{-i}}{\sum}  \lambda_u 
  + \underset{u \in S_{-i}}{\sum} \lambda_u \right) \\
&= \frac{c a_{-i}}{a\alpha} \left( a_i  \right)
  - \frac{c a_i}{a\alpha} \left( a_{-i} \right) = 0.
 \end{align*}
 
The fact that the resulting payoffs-vector is in the core follows from Theorem \ref{thm:DAGCore} in the Section \ref{appendix} where
we characterize a sufficient condition for the same with \emph{general-aptitudes}. 
This completes the proof.
\end{proof}

Note that in case of perfect proportional-allocation, agents are indifferent
between all available subtasks -- thus herding no longer occurs in this setting.

Finally, as in linear subtask-networks, we can extend the above result to get
conditions for PPS to be the unique equilibrium as shown in the following theorem.

\begin{theorem}
\label{thm:DAGDSIC}
Consider the Treasure-hunt game on a directed acyclic subtask-network $G$, and
the 
SA model with simplicities $\{s_u\}_{u \in [m]}$, and rewards $\{R_u\}_{u \in
[m]}$. 
Define $\forall u \in [m]$, $\beta \triangleq \sum_{i \in L}a_i$, where $L
\in [n]$ is the set of stackelberg players who commit to PPS ex-ante.
Suppose the rewards-vector $\{R_u\}$ satisfies the following
conditions:
\begin{itemize}[nolistsep,noitemsep]
 \item[$\bullet$] There exists a strict total ordering $\prec$ on subtasks $[m]$
such that for every agent $i$ and pair of 
 subtasks $u,v$, we have $v \prec u\iff \frac{R_u s_u}{R_v s_v} >
\frac{a_{-i}}{\beta}$. 
 \item[$\bullet$] For all pairs of subtasks $(u,v)$ such that $v$ is reachable
from $u$ in $G$, we have $v\prec u$.
\end{itemize}
Then the following strategies together constitute a unique Nash
equilibrium:
\begin{itemize}[nolistsep,noitemsep]
\item[$\bullet$] Every agent implements the \emph{PPS} policy.
\item[$\bullet$] At any time, if $M_o\subseteq [m]$ is the set of available
subtasks, 
then every agent $i$ chooses to work on the unique subtask $u^*=\argmax_{u\in
M_o}R_u s_u$
\end{itemize}
as long as agents in $L$ do the above.
\end{theorem}

\begin{proof}
Given an agent $i$, we define $\mathcal{G}_{M,T}^i, (M,T)\in\mathcal{S}$ to be the
subgame
where agent $i$ knows the solutions to subtasks $T^c$, while the Stackelberg
agents $j \in L$
knows the solutions to subtasks $M^c$, i.e., $M^c$ is the set of subtasks that
have been publicly shared collectively.
Note that $T\subseteq M$, and
$T^c\setminus M^c$ are the captive subtasks of $i$, i.e., those which $i$ has
solved, but not publicly shared. Also, for any set $M$ such that
$M^c\in\mathcal{C}$, we define $M_o\subseteq M$ to be the available subtasks in
$M$ according to the dependencies in the graph.

We claim the following invariant: under the
conditions specified in the theorem, PPS is a best response for every 
agent $i \in [n] \setminus L$ in \emph{every} subgame $\mathcal{G}_{M,T}^i$, $(M,T)\in
S$. Using notation
from the proof of Theorem \ref{thm:DAG}, we say subgame
$\mathcal{G}_{M,T}^i\in\mathcal{S}_{l,k}$ if
$(M,T)\in\mathcal{S}_{l,k}$ -- note that this corresponds to agent $i$ having
$k$ unsolved subtasks, while the Stackelberg agents have $l\geq k$ unsolved
subtasks. We can restate the above invariant as follows -- for all $1\leq k\leq
l\leq m$, and for all games $\mathcal{G}_{M,T}^i\in\mathcal{S}_{l,k}$, PPS is
a best response for every agent $i \in [n] \setminus L$. 

For the base case, consider the set of subgames $\mathcal{G}_{M,T}^i\in
\mathcal{S}_{1,1}$, i.e., subgames where all agents have a single open subtask.
Clearly PPS is a best response for $i$ in this case. For the
inductive step, fixing $(l,k)$, we assume the invariant is true for all smaller
subgames:
\begin{itemize}
 \item $\mathcal{G}_{M,T}^i \in \mathcal{S}_{q,p}$ where $q<l$ and $p\leq k$,
 \item $\mathcal{G}_{M,T}^j \in \mathcal{S}_{q,p}$ where $j \neq i$, and $q<l$ and $p\leq
k$.
\end{itemize} 

Similar to the proof of Theorem \ref{thm:DAG}, we need to prove two things assuming our invariant holds for smaller games:
\begin{enumerate}
 \item that ``PPS'' part of the invariant holds for subgames in $\mathcal{S}_{l,k}$ where $k<l$, and 
 \item given that (1) holds, we need to prove the ``herding'' condition holds for subgames of the type $\mathcal{S}_{k,k}$. 
\end{enumerate}

We can prove (1) in a way similar way to the proof of Theorem \ref{thm:LineDSIC}. We will here prove (2).
So effectively, given that the invariant holds for all smaller games, and given ``PPS'', we need to show that, for a given 
agent $i$, his best response in the subgame $\mathcal{G}_{T,T}$ is to choose the subtask $u \triangleq u^*(T)$. For simplicity of analysis,
assume there is just one agent $j \notin L \cup \{i\}$ who is on some subtask $w$(possibly $w=u$).

In the scenario described above, let $U_i(u)$ and $U_i(v)$ denote the utility obtained by $i$ when she chooses 
subtask $u$ and some $v \neq u$ respectively. Following our usual notation, $U^i_{T,T}$ represents $i$'s value of the subgame
$\G_{M,T}$ under ``PPS + herding''.
From similar calculations we did earlier, we have that
\begin{align*}
 U^i(v) - U^i_{T,T} &= \frac{a_i(v)}{a_i(v)+\beta(u)+a_j(w)}\frac{a_{-i}(v)}{a_i(v)+a_{-i}(v)}R_v 
 - \frac{\beta(u)}{a_i(v)+ \beta(u) + a_i(w)}\frac{a_{i}(u)}{a_i(u)+a_{-i}(u)}R_u \\
 &- \frac{a_j(w)}{a_i(v)+ \beta(u) + a_i(w)}\frac{a_{i}(w)}{a_i(w)+a_{-i}(w)}R_w \\
 & \leq \frac{1}{a_i(v)+ \beta(u) + a_i(w)}\left( \frac{a_ia_{-i}}{a_i+a_{-i}}R_vs_v - \frac{a_i\beta}{a_i+a_{-i}}R_us_u \right) \\
 & \leq 0
\end{align*}
\begin{align*}
 U^i(u) - U^i_{T,T} &= \frac{a_i(u)}{a_i(u)+\beta(u) +a_j(w)}\frac{a_{-i}(u)}{a_i(u)+a_{-i}(u)}R_u 
 - \frac{\beta}{a_i(u)+ \beta(u) + a_i(w)}\frac{a_{i}(u)}{a_i(u)+a_{-i}(u)}R_u \\
 &- \frac{a_j(w)}{a_i(u)+ \beta(u) + a_i(w)}\frac{a_{i}(w)}{a_i(w)+a_{-i}(w)}R_w \\
 & = \frac{1}{a_i(u)+\beta(u) +a_j(w)} \left( \frac{a_ia_{-i}}{a_i+a_{-i}}R_us_u - \frac{a_i\beta}{a_i+a_{-i}}R_us_u 
 -\frac{a_i a_j}{a_i+a_{-i}}R_ws_w\right) \\
 &= \frac{1}{a_i(u)+\beta(u) +a_j(w)} \frac{a_i}{a_i+a_{-i}}\left((a_{-i}-\beta)R_u s_u - a_jR_ws_w\right) \\
 &=  \frac{1}{a_i(u)+\beta(u) +a_j(w)} \frac{a_i}{a_i+a_{-i}}\left(a_jR_u s_u - a_jR_ws_w\right) \\
 & \geq 0
 \end{align*}
The above calculations show that $U^i(u) > U^i(v)$ and that clinches the proof.
\end{proof}
\section{The Efficiency of PPS in the Treasure-Hunt Game}
\label{sec:poa}

Having obtained conditions for incentivizing PPS, we next turn to the problem of
minimizing the \emph{expected makespan} of a multi-stage project. The importance
of PPS is clear in linear subtask-networks, where \emph{ensuring PPS is
necessary and sufficient for minimizing the expected makespan}. Essentially,
under PPS, all agents work together on a single current subtask, and thus each
subtask gets solved at a rate equal to the sum of rates of the agents. On the
other hand, as we show in Example \ref{eg:poanopps}, in settings without PPS,
each subtask is solved at a rate close to that of the maximum single-agent rate.

An ideal result would be if ensuring PPS alone is sufficient to minimize
makespan in all acyclic subtask-networks. However the following example shows
this is not possible:

\begin{example}
\label{eg:poawithpps}
Consider a treasure-hunt game with $m$ parallel subtasks and $m$ agents. We
index agents and subtasks with unique labels from $[m]$. Suppose the aptitudes
are given by: $a_i(u)=1$ if $i=u$, else $1/(m-1)$. From Theorem \ref{thm:DAG},
we know that we can set rewards so as to ensure that all agents follow PPS $+$
herding, i.e., they solve the subtasks as a group in a fixed order. Thus we have
that $\EE[T_{PPS}]=\frac{m}{1+(m-1)/(m-1)}=m/2$. On the other hand, suppose
agent $i$ starts on task $u=i$, and stops after finishing it -- clearly this is
an upper bound on the optimal makespan, and so we have
$\EE[T_{OPT}]=\frac{1}{m}+\frac{1}{m-1}+\ldots +1=\Theta(\log m)$. 
\end{example}

Thus, in the equilibrium wherein all agents follow PPS, the expected makespan is
greater than the optimal by a multiplicative factor of $\Omega(m/\log m)$.
Conversely, however, we can show that this is in fact the worst possible ratio
up to a logarithmic factor:

\begin{theorem}
\label{thm:POAgen}
Given an acyclic subtask-network G with edge set $[m]$, suppose the
rewards-vector satisfy the conditions in Theorem \ref{thm:DAG}. Then the
expected makespan of the corresponding equilibrium where all agents follow PPS
is no worse than $m$ times the optimal makespan, irrespective of the number of
agents.
\end{theorem}
\begin{proof}[Proof of Theorem \ref{thm:POAgen}]
As in the proof of Theorem \ref{thm:POA}, given subtask-network $G$, we define
$\mathcal{V}$ as the set of
valid knowledge-subgraphs. Further, for any $M\subseteq [m]$ with
$M^c\in\mathcal{V}$, we defined $M_o\subseteq M$ to be the set of available
tasks. 

Let $OPT$ denote the optimal \emph{centralized} agent-subtask allocation
algorithm. In particular, for any subgame $\mathcal{G}_M$ (i.e.,wherein
solutions for all subtasks in knowledge-subgraph $M^c\in\mathcal{V}$ are
public), $OPT$ chooses the agent-subtask distribution matrix
$\{x^{OPT}_i(u)\}_{i\in[n],u\in M_o}$ that minimizes the expected makespan of
the remaining subtasks. Let $T_{OPT}(M)$ denote the expected time achieved by
$OPT$ in subgame $\mathcal{G}_M$, and for each subtask $u\in M_o$, let
$\lambda^{OPT}_u =\sum_{i \in [n]} x^{OPT}_i(u)a_i(u)$. 

For any arbitrary work-conserving strategy $\{x_i(u)\}_{u\in M_o}$, with
corresponding $\lambda_u$, we claim that subgame $\mathcal{G}_{M}$ reduces to
the subgame $\mathcal{G}_{M-u}$ with probability $\frac{\lambda_u}{\sum_{u\in
M_o}\lambda_u}$, after an expected time of $\frac{1}{\sum_{u\in
M_o}\lambda_u}$. This follows from standard properties
of the exponential distribution. Thus, we have:
\begin{align}
\label{eq:DAG_Topt}	
\EE[T_{OPT}(M)] &= \frac{1}{\sum_{u\in M_o}\lambda^{OPT}_u} + \sum_{u\in M_o}
\frac{\lambda^{OPT}_u}{\sum_{u\in M_o}\lambda^{OPT}_u} \EE[T_{OPT}(M-u)]
\end{align}
This follows from the linearity of expectation, and the fact that $OPT$
chooses the optimal agent-subtask distribution in each subgame.

For a given subgame-network with subtasks $M$, we define the \emph{herding-time}
as:
\begin{align}
\label{eq:DAG_Th}
T_H(M)=\sum_{u\in M}\frac{1}{a(u)}. 
\end{align}
where $a(u) \triangleq \sum_{i\in[n]}a_i(u)$.

Assuming the rewards-vector satisfies the conditions of Theorem \ref{thm:DAG},
the makespan achieved by the 
resulting equilibrium is exactly $T_H(M)$.
Let $\mathcal{S}_k
=\{\mathcal{G}_M :M^c\in\mathcal{V},|M|=k\}$ be the set of all valid subgames
with $k$ remaining subtasks, for $1\leq k \leq m$. We now claim
the following invariant: for every subgraph $M \in \mathcal{S}_k$, we
have $\EE[T_{OPT}(M)] \geq \frac{1}{k} T_H(M)$.

We now prove the above invariant via induction.  For $k=1$, the invariant is
trivially, since the subgames in $S_{1}$ have only a single subtask left. Now
assume the invariant is true for $\fall M\in \mathcal{S}_l$, $l\leq k$. Consider
a subgame $M\in\mathcal{S}_{k+1}$. Applying the induction hypothesis, we have
$\fall u \in M_o,
\EE[T_{OPT}(M-u)] \geq \frac{1}{k}T_H(M-u)$. Now from equations
\ref{eq:DAG_Topt} and \ref{eq:DAG_Th}, we have for $T_{OPT}$:
\begin{align*}
 \EE[T_{OPT}(M)] &= \frac{1}{\sum_{u\in M_o}\lambda^{OPT}_u}  
 + \sum_{u\in M_o} \left[ \frac{\lambda^{OPT}_u}{\sum_{u\in
M_o}\lambda^{OPT}_u} \EE[T_{OPT}(M-u)]\right]\\
&\geq \frac{1}{\sum_{u\in M_o}\lambda^{OPT}_u}  
 + \frac{1}{k}\sum_{u\in M_o} \left[ \frac{\lambda^{OPT}_u}{\sum_{u\in
M_o}\lambda^{OPT}_u} T_{H}(M-u)\right]\\
&= \frac{1}{\sum_{u\in M_o}\lambda^{OPT}_u}  
 + \frac{1}{k}\sum_{u\in M_o} \left[ \frac{\lambda^{OPT}_u}{\sum_{u\in
M_o}\lambda^{OPT}_u} \left( T_{H}(M) - \frac{1}{a(u)}\right) \right] \\
&= \frac{1}{k}T_{H}(M) + \frac{1}{\sum_{u\in M_o}\lambda^{OPT}_u}
\left[ 1 - \frac{1}{k}\sum_{u\in M_o}\frac{\lambda_u^{OPT}}{a(u)}\right].
\end{align*}
Using $c(u) \triangleq \frac{\lambda_u^{OPT}}{a(u)}$, and equation
\ref{eq:DAG_Topt} we get:

\begin{align*}
 \frac{\EE[T_{OPT}(M)]}{T_{H}(M)} &\geq \frac{1}{k} + \frac{1}{\left(\sum_{u\in
M_o}\frac{1}{a(u)}\right)
 \left(\sum_{u\in M_o}c(u)a(u)\right)}
 \left[ 1 - \frac{1}{k}\sum_{u\in M_o}c(u) \right] \\
 &= \frac{1}{k} + \frac{\sum_{u\in M_o}c(u)}{\left(\sum_{u\in
M_o}\frac{1}{a(u)}\right)
 \left(\sum_{u\in M_o}c(u)a(u)\right)}
 \left[ \frac{1}{\sum_{u\in M_o}c(u)} - \frac{1}{k} \right]
\end{align*}
Because $0 \leq c(u) \leq 1$ and $|M_o| \leq k+1$, we have $\sum_{u\in M_o}c(u)
\leq k+1$, and consequently:
\begin{align*}
 \frac{\EE[T_{OPT}(M)]}{T_{H}(M)} &\geq \frac{1}{k} + \frac{\sum_{u\in
M_o}c(u)}{\left(\sum_{u\in M_o}\frac{1}{a(u)}\right)
 \left(\sum_{u\in M_o}c(u)a(u)\right)}
 \left[ \frac{1}{k+1} - \frac{1}{k} \right] \\
 &\geq \frac{1}{k} + \left[ \frac{1}{k+1} - \frac{1}{k} \right] \\
 &= \frac{1}{k+1}.
\end{align*}
This completes the inductive argument, and hence the proof.
\end{proof}

Recall in Example \ref{eg:poanopps}, we show that the expected makespan
\emph{without PPS} can be greater by an $\Omega(n)$ factor. Theorem
\ref{thm:POAgen} shows that under PPS, the makespan can be off by a factor
depending on the number of tasks, but \emph{not the number of agents}. This is
significant in large collaborative projects, where the number of agents often far
exceeds the number of tasks. Moreover, under the SA model, we get a surprising optimality
result:

\begin{theorem}
\label{thm:POA}
Given any acyclic subtask-network $G$ under the SA model, with edge set $[m]$
and agents $[n]$. Then any corresponding equilibrium where all agents follow PPS
also minimizes the expected makespan.
\end{theorem}

In particular, setting rewards to satisfy the conditions in Theorem
\ref{thm:DAGwithSA} \emph{minimizes the makespan}. We now look at the proof of this result.
\begin{proof}[Proof of Theorem \ref{thm:POA}]
As before, given subtask-network $G$, we define $\mathcal{V}$ as the set of
valid knowledge-subgraphs. Further, for any $M\subseteq [m]$ with
$M^c\in\mathcal{V}$, we defined $M_o\subseteq M$ to be the set of available
tasks. Let $OPT$ denote the optimal \emph{centralized} agent-subtask allocation
algorithm. In particular, for any subgame $\mathcal{G}_M$ (i.e.,wherein
solutions for all subtasks in knowledge-subgraph $M^c\in\mathcal{V}$ are
public), $OPT$ chooses the agent-subtask distribution matrix
$\{x^{OPT}_i(u)\}_{i\in[n],u\in M_o}$ that minimizes the expected makespan of
the remaining subtasks. Let $T_{OPT}(M)$ denote the expected time achieved by
$OPT$ in subgame $\mathcal{G}_M$, and for each subtask $u\in M_o$, let
$\lambda^{OPT}_u =\sum_{i \in [n]} x^{OPT}_i(u)a_i$. 

On the other hand, assuming the rewards-vector satisfies the conditions for PPS,
let $T_{PPS}(\mathcal{G}_M)$ denote the makespan achieved by some chosen
equilibrium state in subgame $\mathcal{G}_M$. Also, as we did for $OPT$, we
define $\{x^{eq}_i(u)\}_{i\in[n],u\in M_o}$ (and $\lambda^{eq}_u =\sum_{i \in
[n]} x^{eq}_i(u)a_i$) to be the agent-subtask distribution matrix for some
equilibrium state.

For any arbitrary work-conserving strategy $\{x_i(u)\}_{u\in M_o}$, with
corresponding $\lambda_u$, we claim that subgame $\mathcal{G}_{M}$ reduces to
the subgame $\mathcal{G}_{M-u}$ with probability $\frac{\lambda_us_u}{\sum_{u\in
M_o}\lambda_us_u}$, after an expected time of $\frac{1}{\sum_{u\in
M_o}\lambda_us_u}$. This follows from the SA model, and also standard properties
of the exponential distribution. Thus, we have:
\begin{align}
\label{eq:Topt}	
\EE[T_{OPT}(M)] &= \frac{1}{\sum_{u\in M_o}\lambda^{OPT}_us_u} + \sum_{u\in M_o}
\frac{\lambda^{OPT}_us_u}{\sum_{u\in M_o}\lambda^{OPT}_us_u} \EE[T_{OPT}(M-u)]\\
\label{eq:Tpps}	
\EE[T_{PPS}(\mathcal{G}_M)] &= \frac{1}{\sum_{u\in M_o}\lambda^{eq}_us_u} +
\sum_{u\in M_o} \frac{\lambda^{eq}_us_u}{\sum_{u\in M_o}\lambda^{eq}_us_u}
\EE[T_{PPS}(\mathcal{G}_{M-u})].
\end{align}
The first follows from linearity of expectation, and the property that $OPT$
chooses the optimal agent-subtask distribution in each subgame; the second
follows from the subgame-perfect property of the equilibrium.

For a given subgame-network with subtasks $M$, we define the \emph{herding-time}
$T_H(M)=\sum_{u\in M}\frac{1}{As_u}$, where $A=\sum_{i\in[n]}a_i$. We now claim
the following invariant: for every subgraph $M$ with $M^c\in\mathcal{V}$, we
have $\EE[T_{PPS}(\mathcal{G}_M)]=\EE[T_{OPT}(M)]=T_H(M)$. Since this is true
for all subgames $\mathcal{G}_M$, therefore we have that the price-of-anarchy is
$1$ for all acyclic subtask networks under the SA model.

We now prove the above invariant via induction. Let $\mathcal{S}_k
=\{\mathcal{G}_M :M^c\in\mathcal{V},|M|=k\}$ be the set of all valid subgames
with $k$ remaining subtasks, for $1\leq k \leq m$. For $k=1$, the invariant is
trivially, since the subgames in $S_{1}$ have only a single subtask left. Now
assume the invariant is true for $\fall M\in \mathcal{S}_l$, $l\leq k$. Consider
a subgame $\mathcal{M}\in\mathcal{S}_{k+1}$. Applying the induction hypothesis
to equations \ref{eq:Topt} and \ref{eq:Tpps}, we have $\fall u \in M_o,
\EE[T_{OPT}(M-u)] = \EE[T_{PPS}(\mathcal{G}_{M-u})]=T_H(M-u)$. Now from equation
\ref{eq:Topt}, we have for $T_{OPT}$:
\begin{align*}
\EE[T_{OPT}(M)] &= \frac{1}{\sum_{u\in M_o}\lambda^{OPT}_us_u }  
 + \sum_{u\in M_o} \left[ \frac{\lambda^{OPT}_us_u}{\sum_{u\in
M_o}\lambda^{OPT}_us_u} \EE[T_{H}(M-u)]\right]\\
&= \frac{1}{\sum_{u\in M_o}\lambda^{OPT}_us_u }+\sum_{u\in M_o}
\left[\frac{\lambda^{OPT}_us_u}{\sum_{u\in
M_o}\lambda^{OPT}_us_u}\left(\EE[T_{H}(M)]-\frac{1}{As_u} \right)\right].
\end{align*}
Simplifying this expression, we get:
\begin{align*}
\EE[T_{OPT}(M)]
= \EE[T_{H}(M)] + \frac{1}{\sum_{u\in M_o}\lambda^{OPT}_us_u } -
\frac{1}{\sum_{u\in M_o}\lambda^{OPT}_us_u} 
\left(\frac{\sum_{u\in M_o} \lambda^{OPT}_u}{A} \right) = \EE[T_{H}(M)].
\end{align*}
Similarly, we can do the same for $T_{PPS}$, using equation \ref{eq:Tpps} and
the induction hypothesis. Combining the two, we get the desired result.
\end{proof}

\section{Discussion}\label{sec:discussion}
As mentioned in Section \ref{sec:intro}, in a stylized model like ours, some phenomena may arise from modeling artifacts. We now briefly argue that our results are mostly insensitive to our assumptions.

We focus on a setting with \emph{compulsory participation} -- agents derive utility only by solving subtasks. A vast majority of research is done by professional researchers, committed to working on their chosen topics, which makes the act of research different from settings such as crowdsourcing \cite{chawla2012optimal}. Moreover, our results are not sensitive to the number of researchers engaged in a project -- the question of partial-progress sharing arises as soon as there is more than one participant.

We assume that the time an agent takes to solve a task has an \emph{exponential distribution}. Such an assumption is common when considering dynamics in games as it helps simplify the analysis (for example, see \cite{ghosh2013incentivizing}). We conjecture however that the equilibria in the treasure-hunt game remain qualitatively similar under a much broader class of distributions. One reason for this is that assuming exponentially distributed completion times is, in a sense, \emph{over-optimistic}. For an agent trying to
decide whether or not to share a breakthrough, assuming memoryless completion times essentially corresponds to assuming other agents are starting from scratch, discounting any time already invested in the problem. 

We assume that agents work on tasks independently. This allows us to focus on information sharing. However, as discussed in Section \ref{ssec:extensions}, our results extend to give conditions for the payoffs-vector resulting from everyone following PPS to be in the core. This might not work if coalitions have rates which are super-additive. Essentially, for PPS to be in the core, we only need the grand coalition to have the fastest rate among all coalitions (see Corollary \ref{cor:SA_core}).

Partial-progress sharing is one of several primitives required for enabling true collaborative research. However, there are very few good formal ways for reasoning about such settings. The model we develop for studying information sharing in collaborative research, though stylized, exhibits a variety of phenomena which correspond to observed behavior among researchers. It captures the uncertain dynamics of research, while retaining analytic tractability. Although we focus on understanding incentives for PPS, our model may prove
useful for understanding other related questions.
\section{Additional theorems}\label{appendix}
\begin{theorem}
\label{thm:LineDSIC}
Define $\forall u \in [m]$, $\beta(u) \triangleq \sum_{i \in L}a_i(u)$, where $L
\in [n]$ is the set of stackelberg players who
commit to PPS ex-ante.
Suppose the rewards-vector $\{R_u\}$ satisfies the following: for any agent $i
\in [n] \setminus L$ and for any pair of 
subtasks $u,v$ such that
$u$ precedes $v$, the rewards satisfy:
\begin{align*}
\frac{R_u \beta(u)}{R_va_{-i}(v)}\geq \frac{a_i(v)}{a_i(v)+a_{-i}(v)}, 
\end{align*}
where for any task $w$, we define $a_{-i}(w)\triangleq\sum_{j\neq i}a_j(w)$.
Then all agents following 
partial-progress sharing (PPS) is the unique Nash equilibrium.
\end{theorem}

\begin{proof}
We number the subtasks from the end, with the last subtask being denoted as $1$,
and the first as $m$. Fix an agent $i\in[n]$ and define $\mathcal{G}_{k,l}^i,
0\leq k\leq l\leq
m,$ to be the subgame where agent $i$ starts at subtask $k$, and the
Stackelberg 
agents $L$ start at subtask $l$. And we formulate the
following invariant: under the conditions specified in the theorem, PPS is a
best response for every agent $j \in [n]$  in \emph{every} subgame
$\mathcal{G}_{k,l}^j, 1\leq
k\leq l\leq m$. 

Given two subgames $\mathcal{G}_{p,q}^j$
and $\mathcal{G}_{k,l}^i$, we say that $\mathcal{G}_{p,q}^j$ is smaller than
$\mathcal{G}_{k,l}^i$ if:
\begin{itemize}
 \item $i=j ~\land~ (p<k ~\vee~ (p=k~ \land~ q<l))$, or 
 \item $i \neq j ~\land~ q<l $.
\end{itemize}

Given a subgame $\mathcal{G}_{k,l}^i$,
assuming the above-mentioned invariant is true for all smaller subgames, agent
$i$ is certain 
that there is no other agent ahead of her. For, if there were an agent $j$ ahead
of her at subtask
$r<k$, agent $j$ would follow PPS by virtue of being in a smaller game.

And so, the expected reward earned by agent $i$ by following PPS in the subgame
$\mathcal{G}_{k,l}^i$ is given by:
\begin{align*}
U_{k,l}^i=\sum_{u=k+1}^{l}R_u+\sum_{u=1}^{k}R_u.\left(\frac{a_i(u)}
{a_i(u)+a_{-i}(u)}\right),	
\end{align*}
where $a_{-i}(u)=\sum_{j\in[n]\setminus\{i\}}a_j(u)$. 

As discussed in the proof of Theorem \ref{thm:DAG}, given a subgame
$\mathcal{G}_{k,l}^i$, we 
need only show that, for agent $i$ following PPS dominates one step deviations,
i.e., waiting
to progress to subtask $k-1$ before sharing any progress (unless, as we have by
the invariant, some
other agent who is behind declares some progress).

Assume that agent $i$ is deviating as explained above. For any strategy profile
of agents in 
$[n]\setminus(L \cup \{i\})$, let $p_{k-1}$ and $p_{l^\prime}$ denote the
probabilities that 
$\mathcal{G}_{k,l}^i$ devolves to  $\mathcal{G}_{k-1,k-1}^i$ and $\mathcal{G}_{k,l^\prime}^i$, where $k \leq
l^\prime < l$, when any
of the above agents declares progress upto a subgame between $l$ and $k-1$.
Since the invariant
dictates that $i$ will follow PPS in each of these subgames, the game
transitions to $\mathcal{G}_{k,k}^j$ for all
$j \in [n]$. And the first agent $j ^ \prime$ to make progress will move to
$\mathcal{G}_{k-1,k}^{j^\prime}$, and by the
invariant will follow PPS. So effectively, assuming the invariant ensures that
every agent follows PPS in
the ensuing subgame. So the reward gained by agent $i$ in $\mathcal{G}_{k,l^\prime}^i$ is:
\begin{align} \label{eqn:lineDSIC1}
U_{k,l^\prime}^i=\sum_{u=k+1}^{l^\prime}R_u+\sum_{u=1}^{k}R_u.\left(\frac{a_i(u)
}
{a_i(u)+a_{-i}(u)}\right) \leq U_{k,l-1}^i < U_{k-1,l}^i.	
\end{align}
And similarly, 
\begin{align} \label{eqn:lineDSIC2}
U_{k-1,k-1}^i < U_{k,l-1}^i.  
\end{align}

Summing up, since the Stackelberg agents $L$ are on subtask $l$ in $\mathcal{G}_{k,l}^i$, we have
that $\mathcal{G}_{k,l}^i$ transitions to:
\begin{itemize}
 \item $\mathcal{G}_{k-1,k-1}^i$ with probability $p_{k-1}$,
 \item $\mathcal{G}_{k-1,l}^i$ with probability  $\frac{a_i(k)}{\beta(l)+a_i(k)}(1-p)$, 
 \item $\mathcal{G}_{k,l^\prime}^i$ with probability  $p_{l^\prime}$, for $k \leq l^\prime < l-1$, and
 \item $\mathcal{G}_{k,l-1}^i$ with probability $p_{l-1}+\frac{\beta(l)}{\beta(l)+a_i(k)}(1-p)$, 
\end{itemize}
where $p \triangleq \sum_{k-1 \leq l^\prime < l}p_{l^\prime}$. 

From the above arguments, and using equations
(\ref{eqn:lineDSIC1},\ref{eqn:lineDSIC2}) we have that:
\begin{align*}
U_{k,l}^{i,dev} &=  p_{k-1}U_{k-1,k-1}^i +
\frac{a_i(k)}{\beta(l)+a_i(k)}(1-p)U_{k-1,l}^i +
\sum_{k \leq l^\prime<l-1}p_{l^\prime}U_{k,l^\prime}^i
+ \left( p_{l}+\frac{\beta(l)}{\beta(l)+a_i(k)}(1-p) \right) U_{k,l-1}^i, \\
&\leq  p U_{k,l-1}^i + (1-p)\left( \frac{a_i(k)}{\beta(l)+a_i(k)}U_{k-1,l}^i +
\frac{\beta(l)}{\beta(l)+a_i(k)}U_{k,l-1}^i \right), \\
&\leq \frac{a_i(k)}{\beta(l)+a_i(k)}U_{k-1,l}^i +
\frac{\beta(l)}{\beta(l)+a_i(k)}U_{k,l-1}^i, \\
&= \frac{a_i(k)}{\beta(l)+a_i(k)}\left (U_{k,l}^i + R_k -
\frac{a_i(k)}{a_i(k)+a_{-i}(k)}R_k \right) 
+ \frac{\beta(l)}{\beta(l)+a_i(k)}\left( U_{k,l}^i - R_l \right), \\
&= U_{k,l}^i + \frac{a_i(k)}{\beta(l)+a_i(k)}
\frac{a_{-i}(k)}{a_i(k)+a_{-i}(k)}R_k - \frac{\beta(l)}{\beta(l)+a_i(k)}R_l.
\end{align*}

Subtracting this from the expected utility under PPS, we get:
\begin{align*}
 U_{k,l}^{i} - U_{k,l}^{i,dev} &\geq \frac{\beta(l)}{\beta(l)+a_i(k)}R_l - 
 \frac{a_i(k)}{\beta(l)+a_i(k)} \frac{a_{-i}(k)}{a_i(k)+a_{-i}(k)}R_k \\
 &= \frac{a_{-i}(k)R_k}{\beta(l)+a_i(k)}\left( \frac{\beta(l)R_l}{a_{-i}(k)R_k}
- \frac{a_i(k)}{a_i(k)+a_{-i}(k)} \right) \geq 0
\end{align*}
The proof follows from a backward induction argument, since the stated invariant
is true for all subgames of the form $\mathcal{G}_{0,l}^i$ (and
$\mathcal{G}_{1,1}^i$) for all $i\in [n]$.
\end{proof}

\begin{theorem}\label{thm:DAGCore}
Consider the Treasure-hunt game on a directed acyclic subtask-network $G$, and a
general agent-subtask aptitude matrix $\{a_i(u)\}$. Suppose the rewards-vector
$\{R_u\}$ satisfies the following conditions:
\begin{itemize}[nolistsep,noitemsep]
 \item[$\bullet$] (Monotonicity) There exists a total ordering $\prec$ on
subtasks $[m]$ such that for every agent $i$ and pair of subtasks $u,v$, we have
$v\prec u\iff R_v a_i(v) < R_u a_i(u)$. 
 \item[$\bullet$] For all pairs of subtasks $(u,v)$ such that $v$ is reachable
from $u$ in $G$, we have $v\prec u$.
\end{itemize}
Then the following strategies together constitute a Nash equilibrium:
\begin{itemize}[nolistsep,noitemsep]
\item[$\bullet$] Every agent implements the \emph{PPS} policy.
\item[$\bullet$] At any time, if $M_o\subseteq [m]$ is the set of available
subtasks, then every agent $i$ chooses to work on the unique subtask
$u^*=\argmax_{u\in M_o} R_u a_i(u)$
\end{itemize}
More importantly, the corresponding payoffs are in the core.
\end{theorem}
\begin{proof}
 Clearly, the condition on rewards required above is a special case of the more general condition 
 in Theorem \ref{thm:DAG}. Therefore, ``PPS + herding'' constitutes a Nash equilibrium.
 For each agent $i \in [n]$, the corresponding payoff is given by:
 $u_i \triangleq \underset{u \in [m]}\sum \frac{a_i(u)}{a_i(u)+a_{-i}(u)}R_u$.
 
 Given a coalition $C \subset [n]$, let $\nu(C)$ denote the value of the game for $C$. 
 Define $\forall u \in [m]$, $a_C(u) \triangleq \sum_{i \in C}a_i(u)$ and $a_{-C}(u)\triangleq \sum_{i \in [n] \setminus C}a_i(u)$.  
 By the standard properties of exponential random variables,
 any strategy profile of agents in $C$, i.e., which task each agent works on, can be thought of a mixed strategy of a
 single player of ability $\{a_C(u)\}_{u \in [m]}$. 
 
 Assume every agent in $[n] \setminus C$ does ``PPS + herding''. By the monotonicity condition, we have,
 \begin{align*}
  R_v a_i(v) < R_u a_i(u) &\implies R_v a_C(v) < R_u a_C(u) \mbox{ and } R_v a_{-C}(v) < R_u a_{-C}(u) \\
  &\implies  \frac{R_v a_C(v) a_{-C}(v)}{a_C(v) + a_{-C}(v)} < \frac{R_u a_C(u) a_{-C}(u)}{a_C(u) + a_{-C}(u)}.
 \end{align*}
By Theorem \ref{thm:DAG}, ``PPS + herding'' is a best response for $C$. 
The payoff obtained by $C$ in this scenario is 
\begin{align*}
 \underset{u \in [m]}\sum \frac{a_C(u)}{a_C(u)+a_{-C}(u)}R_u = 
\sum_{i \in C}\underset{u \in [m]}\sum \frac{a_i(u)}{a_i(u)+a_{-i}(u)}R_u = \sum_{i \in C}u_i.
\end{align*}
Therefore, $\nu(C) \leq \sum_{i \in C}u_i$, and the result follows.
\end{proof}
 
\section*{Acknowledgments}
 The authors were supported in part by the DARPA GRAPHS program and the DARPA XDATA program, via grant FA9550-12-1-0411 from the U.S. Air Force Office of Scientific Research (AFOSR) and the Defense Advanced Research Projects Agency (DARPA).

We thank Ramesh Johari and other members of the SOAL Lab for their very useful suggestions.

\end{doublespacing}


\end{document}